\documentclass{article}
\usepackage[utf8]{inputenc}

\title{Breaking the $2^n$ barrier for $5$-coloring and $6$-coloring}
\author{Or Zamir\footnote{Blavatnik School of Computer Science, Tel Aviv University, Israel.}}
\date{}

\usepackage{float}
\usepackage[numbers]{natbib}
\usepackage{graphicx}
\usepackage{comment}
\usepackage{amsmath}
\usepackage{amsthm}
\usepackage{amssymb}
\usepackage{fullpage}
\usepackage{bbold}
\usepackage{mathabx}
\usepackage[noend]{algorithm2e}

\newtheorem{theorem}{Theorem}[section]

   \newtheorem{lemma}[theorem]{Lemma}
   \newtheorem{corollary}[theorem]{Corollary}
   \newtheorem{claim}[theorem]{Claim}
   
   \newtheorem{definition}[theorem]{Definition}

\newtheorem{open}{Open Problem}
\newtheorem*{open*}{Open Problem}
   
   \newtheorem{observation}{Observation}
   
\renewcommand{\leq}{\leqslant}

\renewcommand{\geq}{\geqslant}

\newcommand\restr[2]{{
  \left.\kern-\nulldelimiterspace 
  #1 
  \vphantom{\big|} 
  \right|_{#2} 
  }}

\begin{document}

\maketitle

\begin{abstract}
The coloring problem (i.e., computing the chromatic number of a graph) can be solved in $O^*(2^n)$ time, as shown by Bj{\"o}rklund, Husfeldt and Koivisto in 2009.
For $k=3,4$, better algorithms are known for the $k$-coloring problem.
$3$-coloring can be solved in $O(1.33^n)$ time (Beigel and Eppstein, 2005) and $4$-coloring can be solved in $O(1.73^n)$ time (Fomin, Gaspers and Saurabh, 2007).
Surprisingly, for $k>4$ no improvements over the general~$O^*(2^n)$ are known.
We show that both $5$-coloring and $6$-coloring can also be solved in $O\left(\left(2-\varepsilon\right)^n\right)$ time for some~$\varepsilon>0$.
As a crucial step, we obtain an exponential improvement for computing the chromatic number of a very large family of graphs.
In particular, for any constants~$\Delta,\alpha>0$, the chromatic number of graphs with at least~$\alpha\cdot n$ vertices of degree at most~$\Delta$ can be computed in $O\left(\left(2-\varepsilon\right)^n\right)$ time, for some~$\varepsilon = \varepsilon_{\Delta,\alpha} > 0$.
This statement generalizes previous results for bounded-degree graphs (Bj{\"o}rklund, Husfeldt, Kaski, and Koivisto, 2010) and graphs with bounded average degree (Golovnev, Kulikov and Mihajilin, 2016). 
We generalize the aforementioned statement to List Coloring, for which no previous improvements are known even for the case bounded-degree graphs.
\end{abstract}

\section{Introduction}
The problem of $k$-coloring a graph, or determining the \emph{chromatic number} of a graph (i.e., finding the smallest $k$ for which the graph is $k$-colorable) is one of the most classic and well studied NP-Complete problems.
Computing the chromatic number is listed as one of the first NP-Complete problems in Karp's paper from 1972 \cite{karp1972reducibility}.
In a similar fashion to $k$-SAT, the problem of $2$-coloring is polynomial, yet $k$-coloring is NP-complete for every $k\geq 3$ (proven independently by Lov{\'a}sz \cite{lovasz1973coverings} and Stockmeyer \cite{stockmeyer1973planar}).
An algorithm solving $3$-coloring in sub-exponential time would imply, via the mentioned reductions, that $3$-SAT can also be solved in sub-exponential time. 
It is strongly believed that this is not possible (as stated in a widely believed conjecture called \emph{The Exponential Time Hypothesis} \cite{impagliazzo2001complexity}), and thus it is believed that exact algorithms solving $k$-coloring must be exponential. 

There is a substantial and ever-growing body of work exploring exponential-time worst-case algorithms for NP-Complete problems. A 2003 survey of Woeginger \cite{woeginger2003exact} covers and refers to dozens of papers exploring such algorithms for many problems including satisfiability, graph coloring, knapsack, TSP, maximum independent sets and more.
Subsequent review article of Fomin and Kaski \cite{fomin2013exact} and book of Fomin and Kratsch \cite{fomin2010exact} further cover the topic of exact exponential-time algorithms.

For satisfiability (commonly abbreviated as SAT), the running time of the trivial algorithm enumerating over all possible assignments is $O^*(2^n)$.
No algorithms solving SAT in time $O^*\left(\left(2-\varepsilon\right)^n\right)$ for any $\varepsilon>0$ are known, and a popular conjecture called The \emph{Strong} Exponential Time Hypothesis \cite{calabro2009complexity} states that no such algorithm exists. 
On the other hand, it is known that for every fixed $k$ there exists a constant $\varepsilon_k>0$ such that $k$-SAT can be solved in $O^*\left(\left(2-\varepsilon_k\right)^n\right)$ time.
A result of this type was first published by Monien and Speckenmeyer in 1985 \cite{monien1985solving}. 
A long list of improvements for the values of $\varepsilon_k$ were published since, including the celebrated 1998 PPSZ algorithm of Paturi, Pudl{\'a}k, Saks and Zane \cite{paturi2005improved} and the recent improvement over it by Hansen, Kaplan, Zamir and Zwick \cite{hansen2019faster}.

For coloring, on the other hand, the situation is less understood.
The trivial algorithm solving $k$-coloring by enumerating over all possible colorings takes $O^*(k^n)$ time.
Thus, it is not even immediately clear that computing the chromatic number of a graph can be done in $O^*(c^n)$ time for a constant $c$ independent of~$k$.
In 1976, Lawler \cite{lawler1976note} introduced the idea of using dynamic-programming to find the minimal number of independent sets covering the graph. 
The trivial implementation of this idea results in an $O^*(3^n)$ algorithm. 
More sophisticated bounds on the number of maximal independent sets in a graph and fast algorithms to enumerate over them (Moon and Moser \cite{moon1965cliques}, Paull and Unger \cite{paull1959minimizing}) resulted in an $O^*(2.4422^n)$ algorithm. 
This was improved several times (including Eppstein \cite{eppstein2001small} and Byskov \cite{byskov2004enumerating}), until finally an algorithm computing the chromatic number in $O^*(2^n)$ time was devised by Bj{\"o}rklund, Husfeldt and Koivisto in 2009 \cite{bjorklund2009set}. This settled an open problem of Woeginger \cite{woeginger2003exact}.
A relatively recent survey of Husfeldt \cite{husfeldt_2015} covers the progress on graph coloring algorithms.

For $k=3,4$, better algorithms are known for the $k$-coloring problem.
Schiermeyer \cite{schiermeyer1993deciding} showed that $3$-coloring can be solved in $O^*(1.415^n)$ time. 
Biegel and Eppstein \cite{beigel20053} gave algorithms solving $3$-coloring in $O^*(1.3289^n)$ time and $4$-coloring in $O^*(1.8072^n)$ time in 2005. 
Fomin, Gaspers and Saurabh \cite{fomin2007improved} have improved the running time of $4$-coloring to $O^*(1.7272^n)$ in 2007.
Unlike the situation in $k$-SAT, for every $k>4$ the best known running time for $k$-coloring is $O^*(2^n)$, the same as computing the chromatic number. 
Thus, a very fundamental question was left wide open.

\begin{open}\label{open5}
Can $5$-coloring be solved in $O^*\left(\left(2-\varepsilon \right)^n\right)$ time, for some $\varepsilon>0$?
\end{open}

More generally,

\begin{open}\label{openk}
Can $k$-coloring be solved in $O^*\left(\left(2-\varepsilon_k \right)^n\right)$ time, for some $\varepsilon_k>0$, for every $k$?
\end{open}

In our work, we answer Problem~\ref{open5} affirmatively, the answer extends to $6$-coloring as well. We also make steps towards settling Problem~\ref{openk}.

The main technical theorem of our paper follows.

\begin{definition}\label{defbounded}
For $0\leq \alpha \leq 1$ and $\Delta>0$ we say that a graph $G=\left(V(G),E(G)\right)$ is  \emph{\mbox{$(\alpha,\Delta)$-bounded}} if it contains at least $\alpha\cdot |V(G)|$ vertices of degree at most $\Delta$.
\end{definition}
\begin{theorem}\label{mainthm}
For every $\Delta,\alpha>0$ there exists $\varepsilon_{\Delta,\alpha}>0$ such that we can compute the chromatic number of \\ \mbox{$(\alpha,\Delta)$-bounded} graphs in $O\left(\left(2-\varepsilon_{\Delta,\alpha}\right)^n\right)$ time.
\end{theorem}

In other words, we can answer Problem~\ref{openk} affirmatively unless the graph has almost only vertices of super-constant degrees. This theorem generalizes a few previous results.
A similar statement for the restricted case of bounded degree graphs was obtained by Bj\"{o}rklund et al. in \cite{bjorklund2010trimmed}. 
Golovnev, Kulikov and Mihajilin \cite{golovnev2016families} used a variant of FFT to get an algorithmic improvement for computing the chromatic number of graphs with bounded average degree.
Prior to that, Cygan and Pilipzcuk \cite{cygan2015faster} obtained an improvement for the running time required for the Traveling Salesman Problem for graphs with bounded average degree.

It is important to stress that Theorem~\ref{mainthm} is much more general than the mentioned results.
In particular, $(\alpha,\Delta)$-bounded graphs may have~$\Omega(n^2)$ edges (and in turn average degree~$\Omega(n)$). The techniques used for graphs with bounded average degree do not extend to this case. 
Another important difference is that our algorithms extend to finding a coloring of the graph while the mentioned ones solve the decision problem, as later discussed in Section~\ref{decvssearch}.
Moreover, the generality of Theorem~\ref{mainthm} is crucial for our derivation of the reductions resulting in the improvements for~$5$ and~$6$ coloring.

In the List Coloring problem (defined formally in Section~\ref{prel}), we are given a graph and lists $C_v$ of colors for each vertex $v$, and are asked to find a valid coloring of the graph such that each vertex $v$ is colored by some color appearing in its list $C_v$. In the $k$-list-coloring problem each list $C_v$ is of size at most $k$.

List Coloring can also be solved in $O^*(2^n)$ time \cite{bjorklund2009set}, yet no improvements were known even for the bounded-degree case.
We extend Theorem~\ref{mainthm} to $k$-list-coloring, for any constant $k$, as follows. 
\begin{theorem}\label{mainthmlst}
For every $k,\Delta,\alpha>0$ there exists $\varepsilon_{k,\Delta,\alpha}>0$ such that we can solve $k$-list-coloring for \\ \mbox{$(\alpha,\Delta)$-bounded} graphs in $O\left(\left(2-\varepsilon_{k,\Delta,\alpha}\right)^n\right)$ time, regardless of the size of the universe of colors.
\end{theorem}

Using Theorem~\ref{mainthm} as a crucial component, we devise the following reductions and corollaries.
\begin{theorem}\label{reduct}
Given an algorithm solving $(k-1)$-list-coloring in time $O\left(\left(2-\varepsilon\right)^n\right)$ for some constant $\varepsilon>0$, we can construct an algorithm solving $k$-coloring in time $O\left(\left(2-\varepsilon'\right)^n\right)$ for some (other) constant $\varepsilon'>0$.
Furthermore, the reduction is deterministic.
\end{theorem}

\begin{theorem}\label{reduct2}
Given an algorithm solving $(k-2)$-list-coloring in time $O\left(\left(2-\varepsilon\right)^n\right)$ for some constant $\varepsilon>0$, we can construct an algorithm solving $k$-coloring with high probability in time $O\left(\left(2-\varepsilon'\right)^n\right)$ for some (other) constant $\varepsilon'>0$.
\end{theorem}

From which we finally conclude the following, answering Problem~\ref{open5} affirmatively.
\begin{theorem}\label{5col}
$5$-coloring can be solved in time $O\left(\left(2-\varepsilon\right)^n\right)$ for some constant $\varepsilon>0$.
\end{theorem}

\begin{theorem}\label{6col}
$6$-coloring can be solved with high probability in time $O\left(\left(2-\varepsilon\right)^n\right)$ for some constant $\varepsilon>0$.
\end{theorem}

We note that our $5$-coloring algorithm is deterministic, while our $6$-coloring algorithm is randomized with an exponentially small one-sided error probability.

As part of our work, we develop a new removal lemma for small subsets. This could be of independent interest. Very roughly, it states that every collection of small sets must have a large sub-collection that can be made pairwise-disjoint by the removal of a small subset of the universe. The exact statement follows.

\begin{theorem}\label{removalthm}
Let $\mathcal{F}$ be a collection of subsets of a universe $U$ such that every set $F\in \mathcal{F}$ is of size $|F|\leq \Delta$. Let $C>0$ be any constant. Then, there exist subsets $\mathcal{F}'\subseteq \mathcal{F}$ and $U'\subseteq U$, such that
\begin{itemize}
    \item $|\mathcal{F}'| > \rho(\Delta,C)\cdot |\mathcal{F}| + C\cdot |U'|$, where $\rho(\Delta,C)>0$ depends only on $\Delta, C$.
    \item The sets in $\mathcal{F}'$ are disjoint when restricted to $U\setminus U'$, i.e., for every $F_1,F_2 \in \mathcal{F}'$ we have $F_1\cap F_2 \subseteq U'$.
\end{itemize}
\end{theorem}

In Appendix~\ref{removallb} we present an upper bound 
for the function $\rho$ appearing in Theorem~\ref{removalthm}.
This upper bound implies that the constant $\varepsilon$ we can obtain using our technique must be very small.

\subsection{Organization of the Paper}
The rest of the paper is organized as follows.
In Section~\ref{prel} we go over the preliminary tools that we use in the paper. 
In Section~\ref{prevalg} we further elaborate on the $O^*(2^n)$ algorithm of Bj\"{o}rklund, Husfeldt and Koivisto for computing the chromatic number of a graph \cite{bjorklund2009set}.

The main algorithmic contribution of the paper appears in Section~\ref{fasteralg}, in which we prove Theorem~\ref{mainthm}.
The section is partitioned into two main parts. In Section~\ref{maxdegsub} we present our ideas in a simpler manner and get a result limited to bounded degree graphs. Then, in Section~\ref{subsecmain}, which is more technically involved, we complete the proof of Theorem~\ref{mainthm}.
In Section~\ref{genlstcol} we extend the result to the proof of Theorem~\ref{mainthmlst}.

As part of Section~\ref{subsecmain}, we use Theorem~\ref{removalthm}, a combinatorial result independent of the algorithmic tools of Section~\ref{fasteralg}. 
The proof of Theorem~\ref{removalthm} appears in Section~\ref{sec:removal}.

In Section~\ref{sec:red} we use Theorem~\ref{mainthm} as a main ingredient in a reduction from $k$-coloring to $(k-1)$-list-coloring. In this section, we prove Theorems~\ref{reduct} and~\ref{5col}.
In Section~\ref{sec:red2} we refine the ideas used in Section~\ref{sec:red} and construct a reduction from $k$-coloring to $(k-2)$-list-coloring. In this section, we prove Theorems~\ref{reduct2} and~\ref{6col}.

We finally conclude the paper and present a few open problems in Section~\ref{conclusions}.

\section{Preliminaries}\label{prel}

The terminology used throughout the paper is standard.
For a graph $G$ we denote by $V(G)$ and $E(G)$ its vertex-set and edge-set, respectively.
Throughout the paper, $n$ is used to denote $|V(G)|$.
For a subset $V'\subseteq V(G)$ we denote by $G[V']$ the sub-graph of $G$ induced by $V'$.
For $v\in V$ we denote by $\deg(v)$ the degree of $v$ in $G$, by $N(v)$ the set of neighbours of $v$, and by $N[v]:=N(V)\cup\{v\}$.

For $0\leq \alpha \leq 1$ and $\Delta>0$ we say that a graph $G=\left(V(G),E(G)\right)$ is  \emph{\mbox{$(\alpha,\Delta)$-bounded}} if it contains at least $\alpha\cdot |V(G)|$ vertices of degree at most $\Delta$. 
Note that if $\alpha=1$ this definition coincides with the standard definition of a bounded degree graph.

In the \emph{$k$-coloring problem}, we are given a graph $G$ and need to decide whether there exists a $k$-coloring $c:V(G)\rightarrow [k]$ of $G$, such that for every $(u,v)\in E(G)$ we have $c(u)\neq c(v)$.
If a graph has a $k$-coloring, we say that it is $k$-colorable.
In the \emph{chromatic number problem}, we are given a graph $G$ and need to compute $\chi(G)$, the minimal integer $k$ for which $G$ is $k$-colorable.

In the \emph{$k$-list-coloring problem}, we are given a graph $G$ and a set $C_v\subseteq U$ of size $|C_v|\leq k$ for every $v\in V(G)$, where $U$ is some arbitrary universe. We need to decide whether there exists a coloring $c:V(G)\rightarrow U$ such that for every $v\in V(G)$ we have $c(v)\in C_v$ and for every $(u,v)\in E(G)$ we have $c(u)\neq c(v)$.

In a general $(a,b)$-CSP (Constraint Satisfaction Problem, see \cite{kumar1992algorithms} or \cite{schoning1999probabilistic} for a complete definition and discussions) we are given a list of \emph{constraints}\footnote{A general constraint on a set $x_1,\ldots,x_b$ of $a$-ary variables is a subset $T$ of the $a^b$ possible assignments in $\{x_1,\ldots x_b\}\rightarrow[a]$. The constraint is satisfied by an assignment $c$, possibly on more variables, if $\restr{c}{\{x_1,\ldots x_b\}} \in T$. }
on the values of subsets of size $b$ of $n$ $a$-ary variables, and need to decide whether there exists an assignment of values to the variables for which all constraints are satisfied.
$k$-coloring and $k$-list-coloring are examples of $(k,2)$-CSP problems. $k$-SAT is an example of a $(2,k)$-CSP problem.

\subsection{Inverse M\"{o}bius Transform}\label{invmob}
Let $U$ be an $n$-element set.
The \emph{Inverse M\"{o}bius} transform (sometimes also called the \emph{Zeta} transform) \cite{rota} maps a function $f:P(U)\rightarrow \mathbb{R}$ from the power-set of $U$ into another function $\hat{f}:P(U)\rightarrow \mathbb{R}$ defined as
\[
\hat{f}(X) = \sum_{Y\subseteq X} f(Y).
\]
Naively, $\hat{f}(X)$ is computed using $2^{|X|}$ additions. Thus, we can compute all values of $\hat{f}$ in a straightforward manner with $O(3^n)$ operations.
Yates' method from 1937 (\cite{knuth1997seminumerical}, \cite{yates1978design}) improves on the above and computes all values of $\hat{f}$ using just $O(n2^n)$ operations.
The resulting algorithm is usually called \emph{the fast m\"{o}bius transform} or \emph{the fast zeta transform} (\cite{bjorklund2009set}, \cite{kennes1992computational}). 
The authors of \cite{bjorklund2010trimmed} and \cite{bjorklund2009set} use the fast Inverse M\"{o}bius Transform to devise algorithms for combinatorial optimization problems such as computing the chromatic and the domatic numbers of a graph. The algorithm of \cite{bjorklund2009set} is summarized in Section~\ref{prevalg}. 

A description of Yates' method follows.

\begin{lemma}
The Inverse M\"{o}bius Tranform $\hat{f}$ for some function $f:P(U)\rightarrow \mathbb{R}$ can be computed in $O(n2^n)$ time, where $n:=|U|$.
\end{lemma}
\begin{proof}
Denote by $U = \{u_1,\ldots,u_n\}$ some enumeration of $U$'s elements.
Denote by $f_0:=f$.
We preform $n$ iterations for $i=1,\ldots, n$, in which we compute all values of the function $f_i : P(U)\rightarrow \mathbb{R}$  defined using $f_{i-1}$ as follows.
\begin{equation*}
  f_i(X) =
    \begin{cases}
      f_{i-1}(X) + f_{i-1}(X\setminus \{u_i\}) & \text{if $u_i \in X$}\\
      f_{i-1}(X) & \text{otherwise}
    \end{cases}       
\end{equation*}

Namely, in the $i$-th iteration we add the values the function gets in the sub-cube defined by $u_i=0$ to the corresponding values in the sub-cube defined by $u_i=1$.

A simple induction on $i$ shows that $f_i(X) = \sum_{Y\in S_i(X)} f(Y)$ where $S_i(X)$ is the set of all subsets $Y\subseteq X$ such that
\[
\{u_j \in Y \; | \; j>i \} = \{u_j \in X \; | \; j>i\}
\]
In particular, by the end of the algorithm $f_n=\hat{f}$.
\end{proof}

\subsection{Decision versus Search}\label{decvssearch}
The $k$-coloring problem can be stated in two natural ways. 
In the first, given a graph $G$ decide whether it can be colored using $k$ colors. 
The the second, given a graph $G$ return a $k$-coloring for it if one exists, or say that no such coloring exists.
A few folklore reductions show that the two problems have the same running time up to polynomial factors. We state one for completeness. Others appear in the survey of \cite{husfeldt_2015}.
\begin{lemma}\label{decsearch}
Let $\mathcal{A}$ be an algorithm deciding whether a graph is $k$-colorable in $O(T(n))$ time. Then, there exists an algorithm $\mathcal{A'}$ that finds a $k$-coloring for $G$, if one exists, in $O^*(T(n))$ time.
\end{lemma}
\begin{proof}
We describe $\mathcal{A'}$.
First, use $\mathcal{A}(G)$ to decide whether $G$ is $k$-colorable, if it returns \emph{False} we return that no $k$-coloring exists.
Otherwise, repeat the following iterative process.
For every pair of distinct vertices $(u,v)\notin E(G)$ that is not an edge of $G$, use $\mathcal{A}\left(G':=\left(V(G),E(G)\cup\{(u,v)\}\right)\right)$ to check whether $G$ stays $k$-colorable after adding $(u,v)$ as an edge. 
If it does, add $(u,v)$ to $E(G)$.
We stop when no such pair $(u,v)$ exists.

The reader can verify that the resulting graph must be a complement of $k$ disjoint cliques, and thus we can easily construct a $k$-coloring.
\end{proof}

A problem comes up while trying to use this type of reductions in the settings of this paper. 
The aforementioned reduction adds edges to the graph, and therefore increases the degrees of vertices.
In particular, we cannot use it (or other similar reductions) in a black-box manner for statements like Theorem~\ref{mainthm}.
The algorithm of \cite{bjorklund2010trimmed} solves the decision version of $k$-coloring for bounded degree graphs, and cannot be trivially converted into an algorithm that finds a coloring.
The algorithms presented in this paper, on the other hand, can be easily converted into algorithms that find a $k$-coloring. This is briefly discussed later in Section~\ref{findingcol}.

\section{Overview of the $O^*(2^n)$ algorithm}\label{prevalg}
In this section we present a summary of Bj\"{o}rklund, Husfeldt and Koivisto's algorithm from~\cite{bjorklund2009set}.
We present a concise variant of their work that applies specifically to the coloring problem. The original paper covers a larger variety of set partitioning problems and thus the description in this section is simpler.

We begin by making the following very simple observation, yielding an equivalent phrasing of the coloring problem.
\begin{observation}
A graph $G$ is $k$-colorable if and only if its vertex set $V(G)$ can be \emph{covered} by $k$ independent sets.
\end{observation}

A short outline of the algorithm follows, complete details appear below.
We need to decide whether $V(G)$ can be covered by $k$ independent sets. In order to do so, we compute the number of independent sets in every induced sub-graph and then use a simple inclusion-exclusion argument in order to compute the number of (ordered) covers of $V(G)$ by $k$ independent sets. We are interested in whether this number is positive. 

\begin{definition}
For a subset $V'\subseteq V(G)$ of vertices, let $i(G[V'])$ denote the number of independent sets in the induced sub-graph $G[V']$.
\end{definition}

We next show that using dynamic programming, we can quickly compute these values.
\begin{lemma}\label{computei}
We can compute the values of $i(G[V'])$ for all $V'\subseteq V$ in $O^*(2^n)$ time.
\end{lemma}
\begin{proof}
Let $v\in V'$ be an arbitrary vertex contained in $V'$.
The number of independent sets in $V'$ that do not contain $v$ is exactly $i(G[V'\setminus\{v\}])$.
On the other hand, the number of independent sets in $V'$ that do contain $v$ is exactly $i(G[V'\setminus N[v]])$.
Thus, we have
\[
i(G[V']) = i(G[V'\setminus\{v\}]) + i(G[V'\setminus N[v]]).
\]
We note that both $V'\setminus\{v\}$ and $V'\setminus N[v]$ are of size strictly less than $|V'|$. Thus, we can compute all $2^n$ values of $i(G[\cdot])$ using dynamic programming processing the sets in non-decreasing order of size.
\end{proof}

Consider the expression
\[
F(G) = \sum_{V'\subseteq V(G)} (-1)^{|V(G)|-|V'|}\cdot  i(G[V'])^k . 
\]
Using the values of $i(G[\cdot])$ computed in Lemma~\ref{computei}, we can easily compute the value of $F(G)$ by directly evaluating the above expression in $O^*(2^n)$ time.

\begin{lemma}\label{sumzero}
Let $S_1\subseteq S_2$ be sets. It holds that

\begin{equation*}
  \sum_{S_1\subseteq S \subseteq S_2} (-1)^{|S|}  =
    \begin{cases}
      0 & \text{if $S_1\neq S_2$}\\
      (-1)^{|S_2|} & \text{if $S_1=S_2$}
    \end{cases}       
\end{equation*}
\end{lemma}
\begin{proof}
If $S_1\subsetneq S_2$ then there exists a vertex $v\in S_2\setminus S_1$. We can pair each set $S_1\subseteq S \subseteq S_2$ with $S \triangle \{v\}$, its symmetric difference with $\{v\}$.
Clearly, in each pair of sets one is of odd size and one is of even size, and thus their signs cancel each other. Therefore, the sum is zero.
In the second case, the claim is straightforward.
\end{proof}

\begin{lemma}\label{sumzerouse}
$F(G)$ equals the number of $k$-tuples $(I_0,\ldots,I_{k-1})$ of independent sets in $G$ such that $V(G) = I_0\cup\ldots\cup I_{k-1}$.
\end{lemma}
\begin{proof}
As $i(G[V'])$ counts the number of independent sets in $G[V']$, raising it to the $k$-th power (namely, $i(G[V'])^k$) counts the number of $k$-tuples of independent sets in $G[V']$. 

Let $(I_0,\ldots,I_{k-1})$ be a $k$-tuple of independent sets in $G$.
It appears exactly in terms of the sum corresponding to sets $V'$ such that $I_0\cup\ldots\cup I_{k-1} \subseteq V' \subseteq V(G)$. Each time this $k$-tuple is counted, it is counted with a sign determined by the parity of $V'$. 
By Lemma~\ref{sumzero}, the sum of the signs corresponding to sets $I_0\cup\ldots\cup I_{k-1} \subseteq V' \subseteq V(G)$ is zero if $I_0\cup\ldots\cup I_{k-1} \neq V(G)$ and one if $I_0\cup\ldots\cup I_{k-1} = V(G)$.
\end{proof}

We conclude with
\begin{corollary}
$F(G)$ can be computed in time $O^*(2^n)$, and $G$ is $k$-colorable if and only if $F(G)>0$.
\end{corollary}

\section{Faster Coloring Algorithms for \mbox{$(\alpha,\Delta)$-bounded} Graphs}\label{fasteralg}
The main purpose of this section is proving Theorem~\ref{mainthm}.

We first outline our approach. 
Let $G$ be a graph with a constant chromatic number $\chi(G)\leq k$.
It is well known that $G$ must contain a large independent set. 
Let $S$ be an independent set in $G$. 
We think of $|S|$ as a constant fraction of $|V(G)|$, when we consider $k$ as a constant. 
Let $c:\left(V(G)\setminus S\right) \rightarrow [k]$ be a $k$-coloring of the induced sub-graph $G[V(G)\setminus S]$.
We say that $c$ can be \emph{extended} to a $k$-coloring of $G$ if there exists a proper $k$-coloring $c':V(G)\rightarrow [k]$ such that $\restr{c'}{V(G)\setminus S} = c$.
For a subset $V' \subseteq V(G)\setminus S$ of vertices, we say that \emph{$c$ does not use the full palette on $V'$} if $|c(V')|<k$, namely, if $c$ does not use all $k$ colors on the vertices of $V'$.
Clearly, a proper $k$-coloring $c$ of $V(G)\setminus S$ can be extended to a proper $k$-coloring of $G$ if and only if $|c(N(s))|<k$ for every $s\in S$. 
Our approach, on a high-level, is to construct an algorithm that finds an extendable $k$-coloring of $V(G)\setminus S$.
We aim to do so in $O\left(2^{|V(G)\setminus S|} \left(2-\varepsilon\right)^{|S|}\right)$ time. 

\begin{figure}[H]
\center
\includegraphics[trim={0cm 4.5cm 10cm 0cm}, clip,scale=0.4,page=1]{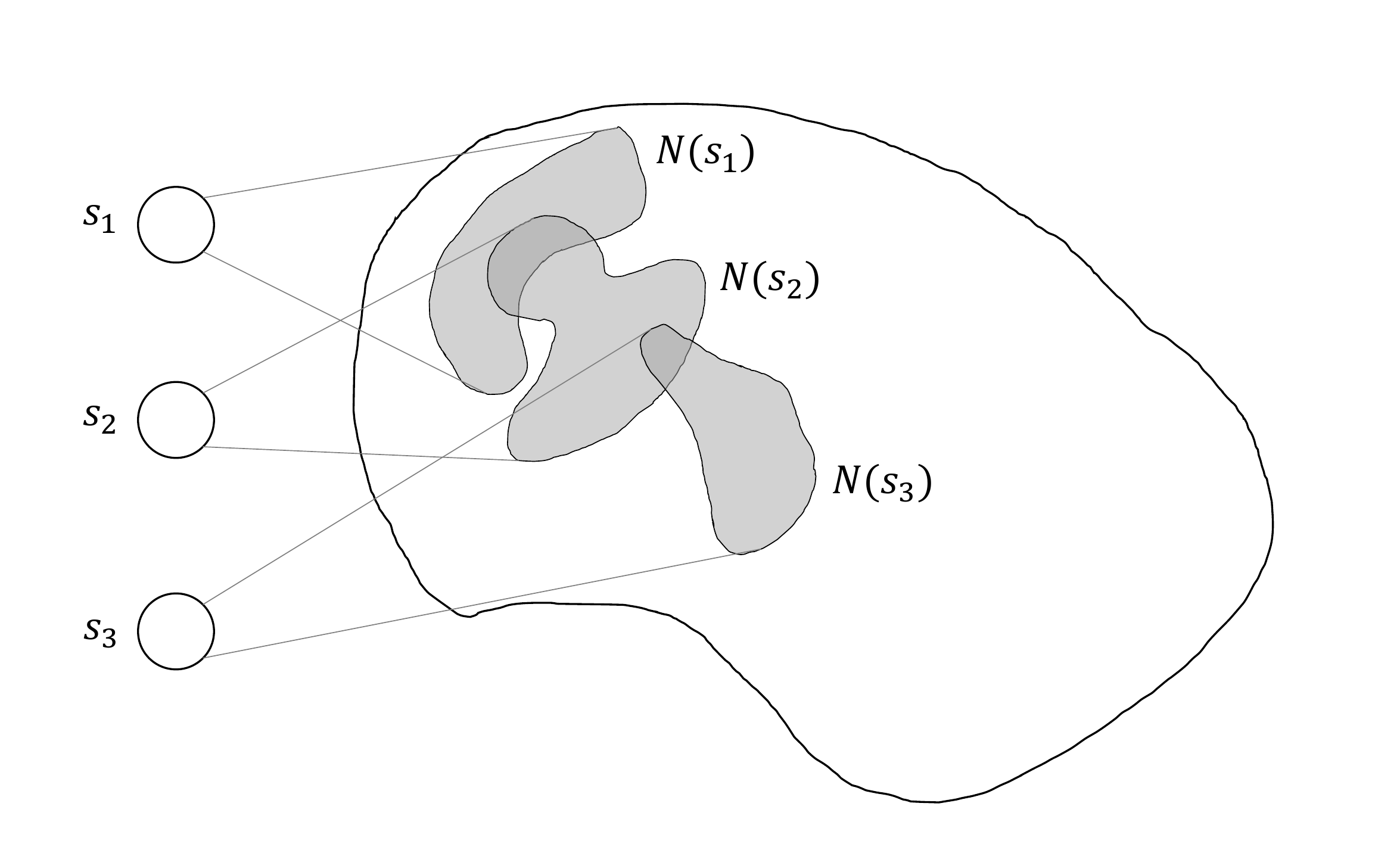}
\label{fig:approach1}
\end{figure}

In Section~\ref{maxdegsub} we consider a restricted version of the problem in which the independent set $S$ has the following two additional properties.
First, we assume that every vertex $s\in S$ is of degree $\deg(s)\leq \Delta$, where $\Delta$ is some constant.
Second, we assume that no pair of vertices $s_1,s_2 \in S$ share a neighbor in $G$.
Equivalently, the neighborhoods $N(s)$ for every $s\in S$ are all disjoint.
Under these conditions, we present an algorithm that runs in $O\left(2^{|V(G)\setminus S|} \left(2-\varepsilon\right)^{|S|}\right)$ time, where $\varepsilon$ depends only on $\Delta$.
As $\varepsilon$ does not depend on $k$, we can in fact compute the chromatic number of $G$ exponentially faster than $O^*(2^n)$ if $G$ contains an independent set $S$ with these properties.
We also observe that if $G$ is of maximum degree $\Delta$ then it contains a large such independent set $S$.
Our algorithm is based on methods that generalize Section~\ref{prevalg}, and on a simple approach to implicitly compute values of the Inverse M\"{o}bius Transform.

\begin{figure}[H]
\center
\includegraphics[trim={0cm 4.5cm 10cm 0cm}, clip,scale=0.4,page=2]{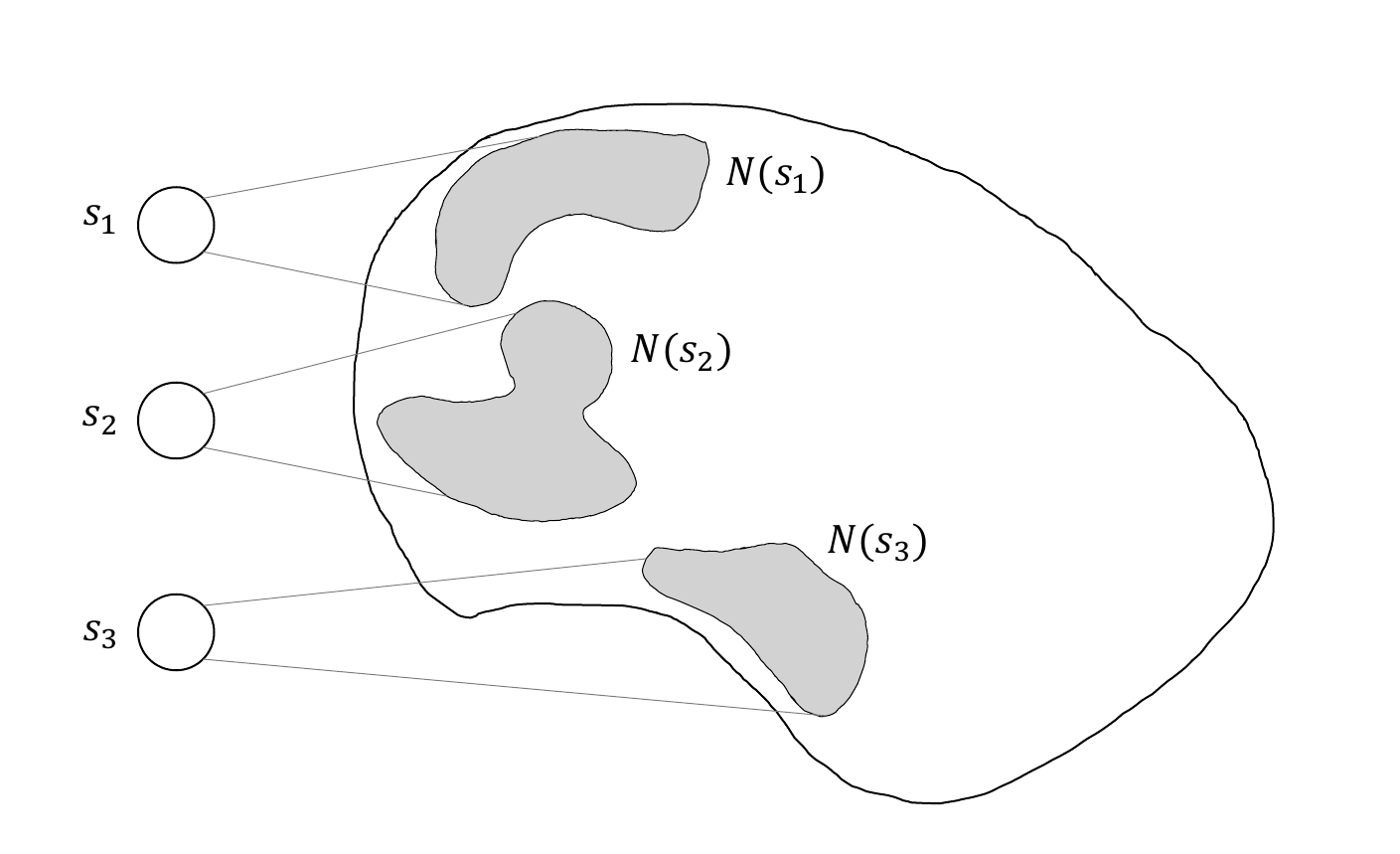}
\label{fig:approach2}
\end{figure}

In Section~\ref{subsecmain} we modify the algorithm of Section~\ref{maxdegsub} and remove the second assumption on $S$. Namely, we now only assume that $S$ is an independent set and that for every $s\in S$ we have $\deg(s)\leq \Delta$.
Our algorithm still runs in $O\left(2^{|V(G)\setminus S|} \left(2-\varepsilon\right)^{|S|}\right)$ time.
A main ingredient in the modification is a new removal lemma for small subsets. The proof of this combinatorial lemma is given in Section~\ref{sec:removal} and its statement is used in a black-box manner in this section.

In Section~\ref{genlstcol} we extend the result to List Coloring.

\subsection{$k$-coloring bounded-degree graphs}\label{maxdegsub}
In this subsection we begin illustrating the ideas leading towards proving Theorem~\ref{mainthm}. We also prove the following (much) weaker statement.
\begin{theorem}\label{maxdegthm}
For every $k,\Delta$ there exists $\varepsilon_{k,\Delta}>0$ such that we can solve $k$-coloring for graphs with maximum degree $\Delta$ in $O\left(\left(2-\varepsilon_{k,\Delta}\right)^n\right)$ time.
\end{theorem}
In fact, as a graph $G$ with maximum degree $\Delta$ has chromatic number $\chi(G)\leq \Delta+1$, we can compute the chromatic number of a graph with degrees bounded by $\Delta$ in time $O\left(\left(2-\varepsilon_{\Delta+1,\Delta}\right)^n\right)$.

As outlined in the beginning of this section, our approach begins by finding a large independent set with some additional properties. 
We show that a graph with bounded degrees must contain a very large independent set $S$ such that the distance between each pair of vertices in $S$ is at least three.
In other words, $S$ is an independent set, and no pair of vertices in $S$ share a neighbor. In particular, the neighborhoods $N(s)$ for $s\in S$ are all disjoint.
The core theorem of this subsection is 
\begin{theorem}\label{easymainthm}
Let $G$ be a graph and $S\subseteq V(G)$ a set of vertices such that the distance between each two vertices in $S$ is at least three and the degree of each vertex in $S$ is at most $\Delta$. For any $k$, we can solve $k$-coloring for $G$ in $O^*\left(2^{|V(G)|-|S|} \cdot (2-2^{-\Delta})^{|S|}\right)$ time.
\end{theorem}
It is important to note that the existence of such a set $S$ is our sole use of the bound on the graph degrees. Note that the bound of Theorem~\ref{easymainthm} does not depend on $k$. Thus, we get an exponential improvement for computing the chromatic number of a graph $G$ that contains a large enough set $S$ with the stated properties. 

Before proving Theorem~\ref{easymainthm}, we describe a simple algorithm for finding a set $S$ with the required properties in bounded-degree graphs.

\begin{lemma}\label{maxdeglemma}
Let $G$ be a graph with maximum degree at most $\Delta$. There exists a set $S\subseteq V(G)$ of at least $\frac{1}{1+\Delta^2}\cdot |V(G)|$ vertices such that the distance between every distinct pair $s_1,s_2\in S$ is at least three. Furthermore, we can find such $S$ efficiently.
\end{lemma}
\begin{proof}
We construct $S$ in a greedy manner.
We begin with $S=\emptyset$ and $V'=V(G)$.
As long as $V'$ is not empty we pick an arbitrary vertex $v\in V'$ and add it to $S$. 
We then remove from $V'$ the vertex $v$ and every vertex of distance at most two from it.

By construction, the minimum distance between a pair of vertices in $S$ is at least three.
The size of the $2$-neighborhood of a vertex is bounded by $1+\Delta+\Delta\cdot (\Delta-1) = 1+\Delta^2$ and thus we get the desired lower bound on the size of $S$.
\end{proof}

Theorem~\ref{maxdegthm} now follows from Lemma~\ref{maxdeglemma} and Theorem~\ref{easymainthm}.

\begin{proof}[Proof of Theorem~\ref{maxdegthm}]
Let $G$ be a graph of maximum degree at most $\Delta$ and let $k$ be an integer.
By Lemma~\ref{maxdeglemma}, we can construct a set $S$ of size $|S|\geq \frac{1}{1+\Delta^2}\cdot |V(G)|$ satisfying the conditions of Theorem~\ref{easymainthm}.
Thus, by Theorem~\ref{easymainthm}, we can solve $k$-coloring for $G$ in time
\[
O^*\left(2^{n-\frac{1}{1+\Delta^2}\;n}\cdot (2-2^{-\Delta})^{\frac{1}{1+\Delta^2}\;n}\right)
=
O^*\left(\left(2\cdot \left(1 - 2^{-\left(\Delta+1\right)}\right)^{\frac{1}{1+\Delta^2}} \right)^{n\;}\right).
\]
\end{proof}

In the rest of the subsection we prove Theorem~\ref{easymainthm}.

\begin{definition}
For subsets $V'\subseteq V(G)\setminus S$ and $S'\subseteq S$ denote by $\beta(V',S')$ the number of independent sets $I$ in $G[V']$ that intersect every neighborhood $N(s)$ of $s\in S'$, that is, $I\cap N(s) \neq \emptyset$ for every $s \in S'$.
\end{definition}

Consider, for a subset $S'\subseteq S$, the following sum
\[
h(G,S') := \sum_{V'\subseteq V(G)\setminus S} (-1)^{|V(G)|-|V'|} \;\beta\left(V', S'\right)^k.
\]

The following proof is almost identical to the proof of Lemma~\ref{sumzerouse} in Section~\ref{prevalg}.

\begin{lemma}\label{meaningofFp}
$h(G,S')$ is the number of covers of $V(G)\setminus S$ by $k$-tuples $(I_0,\ldots,I_{k-1})$ of independent sets in $G[V(G)\setminus S]$ such that $I_i\cap N(s)\neq \emptyset$ for every $s\in S'$ and every $0\leq i \leq k-1$. 
\end{lemma}
\begin{proof}
Each value of $\beta(V',S')$ counts independent sets in $G[V']$ that intersect every neighborhood $N(s)$ for $s\in S'$. 

Each $k$-tuple $(I_0,\ldots,I_{k-1})$ of that type is counted in terms corresponding to sets $V'$ such that
\[
I_0\cup\ldots\cup I_{k-1} \subseteq V' \subseteq V(G)\setminus S.
\]
By Lemma~\ref{sumzero} the multiplicity with which such $k$-tuple is counted is one if 
\[
I_0\cup\ldots\cup I_{k-1} = V(G)\setminus S.
\]
and zero otherwise.
\end{proof}

Consider the following expression.
\[
H(G,S) := \sum_{S'\subseteq S}  (-1)^{|S'|} \; h(G,S')
\]

$H(G,S)$ is the number of covers of $V(G)\setminus S$ by $k$-tuples of independent sets that \emph{do not use the full palette} on any neighborhood $N(s)$ for $s\in S$. 
The precise claim follows.

\begin{lemma}\label{meaningofFptotal}
$H(G,S)$ is the number of covers of $V(G)\setminus S$ by $k$-tuples $(I_0,\ldots,I_{k-1})$ of independent sets in $G[V(G)\setminus S]$ such that for every $s\in S$ there exists $0\leq i \leq k-1$ such that $I_i \cap N(s) = \emptyset$.
\end{lemma}
\begin{proof}
In Lemma~\ref{meaningofFp} we showed that $h(G,S')$ counts the number of covers of $V(G)\setminus S$ by $k$-tuples $(I_0,\ldots, I_{k-1})$ of independent sets in $G[V(G)\setminus S]$ such that for every $s\in S'$ and for every $0\leq i \leq k-1$ we have $I_i\cap N(s) \neq \emptyset$.

A covering $k$-tuple of independent sets $(I_0,\ldots,I_{k-1})$ is counted exactly in terms corresponding to subsets $S'$ such that for every $0\leq i \leq k-1$ and every $s\in S'$, the independent set $I_i$ intersects the neighborhood $N(s)$. These are exactly the subsets $S'$ such that
\[
S' \subseteq \{s\in S\;\;|\;\; \forall 0\leq i\leq k-1,\; I_i\cap N(s) \neq \emptyset \}.
\]

Using Lemma~\ref{sumzero} with $S_1 = \emptyset$ and $S_2 = \{s\in S\;\;|\;\; \forall 0\leq i\leq k-1,\; I_i\cap N(s) \neq \emptyset \}$ we deduce that the multiplicity with which the $k$-tuple is counted is one if 
\[
\{s\in S\;\;|\;\; \forall 0\leq i\leq k-1,\; I_i\cap N(s) \neq \emptyset \} = \emptyset
\]
and zero otherwise.
\end{proof}


As outlined at the beginning of the section, we now claim that $H(G,S)$ is positive if and only if $G$ is $k$-colorable. Note that for the correctness of this lemma we still did not use the disjointness of the neighborhoods $N(s)$. We will need this property to improve the computation time.
\begin{lemma}\label{intuitionformalized}
Let $G$ be a graph and $S$ an independent set in it. Then, $H(G,S)>0$ if and only if $G$ is $k$-colorable.
\end{lemma}
\begin{proof}
Assume that there exists a $k$-coloring $c:V(G)\rightarrow [k]$ of $G$.
For $0\leq i \leq k-1$ denote by 
\[
I_i := \{v\in V(G)\setminus S\; | \; c(v)=i\}
\]
the subset of $V(G)\setminus S$ colored by $i$. 
Each $I_i$ is an independent set as $c$ is a proper coloring of $G$.
Furthermore, for each $s\in S$, the neighborhood $N(s)$ does not intersect $I_{c(s)}$. 
Thus, $(I_0,\ldots,I_{k-1})$ is a cover of $V(G)\setminus S$ by $k$ independent sets that do not all intersect any neighborhood $N(s)$ of $s\in S$.
By Lemma~\ref{meaningofFptotal}, $H(G,S)\geq 1$.

On the other hand, if $H(G,S)>0$ then by Lemma~\ref{meaningofFptotal} there exists a cover by independent sets and in particular a $k$-coloring $c:V(G)\setminus S \rightarrow [k]$ of $G[V(G)\setminus S]$ such that the full palette is not used on any neighborhood $N(s)$ for $s\in S$. 
Thus, we may extend $c$ to a $k$-coloring $c' : V(G) \rightarrow [k]$ of the entire graph by coloring each $s\in S$ with a color that does not appear in $c(N(s))$.
As $S$ is an independent set, this coloring is proper.
\end{proof}

Up to this point, we have formalized the outline from the beginning of this section, reducing $k$-coloring to a problem of $k$-coloring with some restrictions the smaller graph $G[V(G)\setminus S]$ and then to the computation of $H(G,S)$. 

Unfortunately, $H(G,S)$ is a sum of $2^{|S|}$ terms, each of the form $h(G,S')$ which is a sum of $2^{|V(G)|-|S|}$ terms by itself. 
Evidently, there are $2^n$ different terms of the form $\beta(V',S')$ that are used in the definition of $H(G,S)$.
Thus, we cannot hope to compute $H(G,S)$ in less than $2^n$ steps if we need to explicitly examine $2^n$ terms of the form $\beta(\cdot,\cdot)$.
Moreover, it is also not clear how quickly we can compute the values of $\beta(\cdot,\cdot)$.

We begin by explaining how values of $\beta(\cdot)$ can be computed efficiently.
The term $h(G,S')$ is a weighted sum of the values $\beta\left(V', S'\right)$ for all $V' \subseteq V(G)\setminus S$. 
Denote by $\beta_\mu \left(V',S'\right)$ the indicator function that gets the value $1$ if $V'$ is an independent set in $G[V(G)\setminus S]$ and for every $s\in S'$ we have $V' \cap N(s) \neq \emptyset$, and $0$ otherwise.
We can efficiently compute the value of $\beta_\mu$ for a specific input in a straightforward manner (i.e., checking whether it is an independent set that intersects the relevant sets).
We observe that
\[
\beta\left(V', S'\right) = \sum_{V''\subseteq V'} \beta_\mu \left(V'',S'\right),
\]
thus, $\beta = \hat{\beta_\mu}$ as functions of $V'$, and we can compute the values of $\beta\left(V', S'\right)$ for all $V' \subseteq V(G)\setminus S$ in $O^*(2^{|V(G)|-|S|})$ time using the Inverse M\"{o}bius Transform presented in Section~\ref{invmob}.

An improvement to the running time comes from noticing that for many inputs $(V',S')$ the value of $\beta\left(V', S'\right)$ is zero. 
In particular, if $V'\cap N(s)=\emptyset$, for some $s \in S'$, then $\beta\left(V', S'\right) = 0$ as no subset (and in particular no independent set) in $V'$ intersects $N(s)$. 
In the computation of $h(G,S')$ we only need to consider terms corresponding to subsets $V'\subseteq V(G)\setminus S$ in which for every $s\in S'$ the intersection $V'\cap N(s)$ is non-empty, as the values of other terms are all zero.
We present a variant of the Inverse M\"{o}bius Transform that computes only the non-zero values by implicitly setting the others to zero. 
We then show that for most subsets $S' \subseteq S$ the number of non-zero entries is exponentially smaller than $2^{|V(G)|-|S|}$.

\begin{definition}
For any $S'\subseteq S$ denote by $B(S') := \{ V'\subseteq V(G)\setminus S\;\; | \;\; \forall s\in S'.\; V'\cap N(s) \neq \emptyset \}$ the set of all subsets of $V(G)\setminus S$ intersecting all neighborhoods of $S'$.
\end{definition}

As we observed above, for every $V'\notin B(S')$ we have $\beta\left(V', S'\right) = 0$.
We conclude that
\begin{observation}\label{fpzeroobs}
For every $S'$ we have
\[
h(G,S') = \sum_{V'\in B(S')} (-1)^{|V(G)|-|V'|} \;\beta\left(V', S'\right)^k.
\]
\end{observation}

\begin{lemma}\label{fpcompute}
If the neighborhoods $N(s)$ are disjoint for all $s\in S'$, then we can compute $h(G,S')$ in $O^*(|B(S')|)$ time. 
\end{lemma}
\begin{proof}
It suffices to compute $\beta\left(V', S'\right)$ for every $V'\in B(S')$ and then use Observation~\ref{fpzeroobs}.
We do so by introducing a variant of the Inverse M\"{o}bius Transform that implicitly sets the value of $\beta\left(V', S'\right)$ to zero for every $V' \notin B(S')$.

We first note that
\[
B(S') \cong P\left(V(G) \setminus \left(S \cup \bigcup_{s\in S'} N(s)\right) \right) \times \bigtimes_{s\in S'} \left(P(N(s))\setminus \{\emptyset\}\right)
.\]
Thus, we can efficiently construct a simple bijection between $[|B(S')|]$ and $B(S')$ as a Cartesian product. 
We can also efficiently check if a set $V'$ belongs to $B(S')$.
Let $index : B(S') \rightarrow [|B(S')|]$ be a map from $B(S')$ to indices of $[|B(S')|]$. 
If $V'\notin B(S')$ we define $index(V')=-1$. 
By the observation above, we can define $index$ in way for which $index$ and $index^{-1}$ are efficiently computable. We also arbitrarily order the vertices of $V(G)\setminus S$ as $v_1,v_2,\ldots,v_{|V(G)\setminus S|}$.

We describe the algorithm in pseudo-code.

\begin{algorithm}[H]
 Initialize an array $f$ of size $|B(S')|$\;
 \For{$\ell$ in $[|B(S')|]$}
 {
    \eIf{$index^{-1}(\ell)$ is an independent set in $G[V(G)\setminus S]$}
    {
        $f(\ell) \leftarrow 1$ \;
    }
    {
        $f(\ell) \leftarrow 0$ \;
    }
 }
 \For{$i$ in $[|V(G)\setminus S|]$}
 {
    \For{$\ell$ in $[|B(S')|]$}
    {
        $V' \leftarrow index^{-1}(\ell)$ \;
        \If{$v_i \in V'$ and $index(V'\setminus\{v_i\})\neq -1$}
        {
            $f(\ell) \leftarrow f(\ell) + f(index(V'\setminus \{v_i\}))$\;
        }
    }
 }
\end{algorithm}

We view $f$ throughout the algorithm as function $f:B(S')\rightarrow \mathbb{N}$.
Denote the function represented by $f$ at the end of the first \emph{for} loop by $f_0$.
By definition, $f_0(V')=\beta_{\mu}\left(V', S'\right)$ for every $V'\in B(S')$.
Denote by $f_i$ the function represented by $f$ at the end of the $i$-th iteration of the second (outer) \emph{for} loop.

We observe that $f_i$ is defined using $f_{i-1}$ as
\begin{equation*}
  f_i(V') =
    \begin{cases}
      f_{i-1}(V') + f_{i-1}(V'\setminus \{v_i\}) & \text{if $v_i \in V'$}\\
      f_{i-1}(V') & \text{otherwise}
    \end{cases}       
\end{equation*}
where $f_{i-1}(V'\setminus \{v_i\})$ is implicitly defined to be zero if $V'\setminus \{v_i\} \notin B(S')$.

By induction on $i$, similar to this of Section~\ref{invmob}, we can show that
\[
f_i(V') = \sum_{\substack{V''\subseteq V'\\ V''\setminus\{v_1,\ldots,v_i\} = V' \setminus\{v_1,\ldots,v_i\}}} f(V'').
\]
In particular, by the end of the algorithm $f=\hat{f_0}=\hat{\beta_{\mu}}=\beta$ for the entire domain $B(S')$.
\end{proof}

After computing $h(G,S')$ for every $S'\subseteq S$ we can compute $H(G,S)$ in $O^*(2^{|S|})$ time.
We thus finish the proof of Theorem~\ref{easymainthm} with the following counting lemma.

\begin{lemma}\label{sumlemma}
Assume that the neighborhoods $N(s)$ are disjoint for all $s\in S$ and that each neighborhood is of size $|N(s)|\leq \Delta$.
Then, $\sum_{S'\subseteq S} |B(S')| = O^*\left(2^{|V(G)\setminus S|} \cdot (2-2^{-\Delta})^{|S|}\right)$.
\end{lemma}
\begin{proof}
Denote $n(s):=|N(s)|$. Also denote by $N = \bigcup_{s\in S} N(s)$ all neighbors of vertices of $S$ and by $N^c = \left(V(G)\setminus S\right) \setminus N$ their complement in $V(G)\setminus S$.
We have 
\begin{align*}
|B(S')| &= 2^{|N^c|} \cdot \prod_{s\in S'} \left(2^{n(s)}-1\right) \cdot \prod_{s\in S\setminus S'} 2^{n(s)} \\
&= 2^{|N^c|} \cdot \prod_{s\in S'} \left(1-2^{-n(s)}\right) \cdot \prod_{s\in S} 2^{n(s)} \\
&= 2^{|N^c|} \cdot \prod_{s\in S'} \left(1-2^{-n(s)}\right) \cdot 2^{|N|} \\
&= 2^{|V(G)\setminus S|} \cdot \prod_{s\in S'} \left(1-2^{-n(s)}\right).
\end{align*}

For every $s\in S$ we have $n(s)\leq \Delta$ and thus $\left(1-2^{-n(s)}\right) \leq \left(1-2^{-\Delta}\right)$. Hence,
\begin{align*}
|B(S')| &\leq 2^{|V(G)\setminus S|} \cdot \prod_{s\in S'} \left(1-2^{-\Delta}\right) \\
&= 2^{|V(G)\setminus S|} \cdot \left(1-2^{-\Delta}\right)^{|S'|} .
\end{align*}

Therefore we have
\begin{align*}
\sum_{S'\subseteq S} |B(S')| &\leq \sum_{S'\subseteq S} 2^{|V(G)\setminus S|} \cdot \left(1-2^{-\Delta}\right)^{|S'|} \\
&= 2^{|V(G)\setminus S|} \cdot \sum_{i=0}^{|S|} {|S| \choose i} \left(1-2^{-\Delta}\right)^{i} \\
&= 2^{|V(G)\setminus S|} \cdot (2-2^{-\Delta})^{|S|}.
\end{align*}
\end{proof}

\subsection{From bounded-degree graphs to \mbox{$(\alpha,\Delta)$-bounded} graphs}\label{subsecmain}
In this section we prove the main technical theorem of the paper.

\begingroup
\def\thetheorem{\ref{mainthm}}
\begin{theorem}
For every $\Delta,\alpha>0$ there exists $\varepsilon_{\Delta,\alpha}>0$ such that we can compute the chromatic number of~\mbox{$(\alpha,\Delta)$-bounded} graphs in $O\left(\left(2-\varepsilon_{\Delta,\alpha}\right)^n\right)$ time.
\end{theorem}
\addtocounter{theorem}{-1}
\endgroup

We prove the following seemingly weaker statement.
\begin{theorem}\label{weakermainthm}
For every $k,\Delta,\alpha>0$ there exists $\varepsilon_{k,\Delta,\alpha}>0$ such that we can solve $k$-coloring for \\ \mbox{$(\alpha,\Delta)$-bounded} graphs in $O\left(\left(2-\varepsilon_{k,\Delta,\alpha}\right)^n\right)$ time.
\end{theorem}

We then note that Theorem~\ref{weakermainthm} in fact implies Theorem~\ref{mainthm}. Let $G$ be a $(\alpha,\Delta)$-bounded graph. 
We use Theorem~\ref{weakermainthm} for every $1\leq k \leq \Delta$. 
If we did not find a valid coloring of $G$, then $\chi(G)\geq \Delta+1$ and we may use a standard argument (present later in Lemma~\ref{removelowdeg}) to show that removing all vertices of degree at most $\Delta$ does not change $\chi(G)$. 
By definition of $(\alpha,\Delta)$-bounded graphs, removing these vertices leaves a graph with at most $(1-\alpha)n$ vertices and thus the standard chromatic number algorithm runs in $O^*(2^{(1-\alpha)n})$ time.

As in Section~\ref{maxdegsub}, we deduce Theorem~\ref{weakermainthm} from the following theorem.
\begin{theorem}\label{mainthmrefined}
Let $G$ be a graph and $S\subseteq V(G)$ an independent set in $G$. Assume that the degree of each vertex in $S$ is at most $\Delta$. Then, we can solve $k$-coloring for $G$ in 
$O^*\left(2^{|V(G)|} \cdot (1-\varepsilon_{k,\Delta})^{|S|}\right)$
time, for some constant $\varepsilon_{k,\Delta}>0$.
\end{theorem}

Let $G$ be a graph with a subset $U \subseteq V(G)$ of vertices such that for every $v\in U$ we have $\deg(v)\leq \Delta$. 
In a similar fashion to Lemma~\ref{maxdeglemma} of the previous subsection (and even slightly simpler), we can greedily construct a subset $S\subseteq U$ of size $|S|\geq \frac{1}{1+\Delta}\cdot |U|$ which is an independent set. 
Thus, Theorem~\ref{mainthmrefined} immediately implies Theorem~\ref{mainthm}.
Unlike the case of Section~\ref{maxdegsub}, this time the neighborhoods $N(s)$ for $s\in S$ are not necessarily disjoint. Thus, statements comparable to Lemma~\ref{sumlemma} are not true.
Our solution for this problem is surprisingly general. In Section~\ref{sec:removal} we prove the following new type of removal lemma for small sets.

\begingroup
\def\thetheorem{\ref{removalthm}}
\begin{theorem}
Let $\mathcal{F}$ be a collection of subsets of a universe $U$ such that every set $F\in \mathcal{F}$ is of size $|F|\leq \Delta$. Let $C>0$ be any constant. Then, there exist subsets $\mathcal{F}'\subseteq \mathcal{F}$ and $U'\subseteq U$, such that
\begin{itemize}
    \item $|\mathcal{F}'| > \rho(\Delta,C)\cdot |\mathcal{F}| + C\cdot |U'|$, where $\rho(\Delta,C)>0$ depends only on $\Delta, C$.
    \item The sets in $\mathcal{F}'$ are disjoint when restricted to $U\setminus U'$, i.e., for every $F_1,F_2 \in \mathcal{F}'$ we have $F_1\cap F_2 \subseteq U'$.
\end{itemize}
\end{theorem}
\addtocounter{theorem}{-1}
\endgroup

Plugging $\mathcal{F} = \{N(s)\}_{s\in S}$, we get a small set $U' \subseteq V(G)\setminus S$ of graph vertices, and a large subset $S' \subseteq S$ of the independent set, such that the neighborhoods $N(s)$ of $s\in S'$ become pairwise disjoint if we remove the vertices of $U'$ from $G$.
As we want to preserve the correctness of the algorithm, we do not actually remove $U'$ from $G$, but \emph{enumerate} over the colors they receive in a proper $k$-coloring, if one exists.
The main technical gap is adjusting the algorithm and proofs of Section~\ref{maxdegsub} to the case in which some of the graph vertices have fixed colors.

\begin{theorem}\label{fixedthm}
Let $G$ be a graph, $V_0\subseteq V(G)$ a subset of its vertices and $c:V_0 \rightarrow [k]$ a proper $k$-coloring of $G[V_0]$. Denote by $V:=V(G)\setminus V_0$.
Let $S\subseteq V$ be an independent set in $G$ such that the distance in $G[V]$ between each two vertices of $S$ is at least three and the degree in $G[V]$ of each vertex in $S$ is at most $\Delta$. For any $k$, we can decide whether $c$ can be extended to a $k$-coloring of the entire graph $G$ in $O^*\left(2^{|V|-|S|} \cdot (2-2^{-\Delta})^{|S|}\right)$ time.
\end{theorem}

Throughout the rest of the section, it is important to carefully distinguish $V(G)$ from $V$. Note that $V$ does not include the vertices of $V_0$, as their colors are already fixed.
For $j \in [k]$, denote by $V_0^j := c^{-1}(j)$ the subset of $V_0$ colored by $j$ color. Note that $V_0 = \bigcup_{j=1}^{k} V_0^j$.
We begin adapting the algorithm by redefining the $\beta(\cdot,\cdot)$ function.

\begin{definition}
For subsets $V'\subseteq V\setminus S$, $S' \subseteq S$, and a color $j \in [k]$, we denote by $\beta_j(V',S')$ the number of sets $I\subseteq V'$ such that $I\cup V_0^j$ is an independent set in $G$ and that $I\cup V_0^j$ intersects $N(s)$ for every $s\in S'$, that is, for every $s\in S'$ we have $\left(I\cup V_0^j\right) \cap N(s) \neq \emptyset$.
\end{definition}

We also revise the definition of 
\[
h(G,S') := \sum_{V'\subseteq V\setminus S} (-1)^{|V|-|V'|} \; \prod_{j=0}^{k-1} \beta_j \left(V', S'\right).
\]

The proof of Lemma~\ref{meaningofFp} can be easily revised to show the following.
\begin{lemma}\label{newmeaningofFp}
$h(G,S')$ is the number of covers of $V\setminus S$ by $k$-tuples of sets $I_0,\ldots,I_{k-1}$ such that for every $j\in[k]$, $I_j\cup V_0^j$ is an independent set in $G$ and that for every $s\in S'$ and every $j\in [k]$ the set $I_j\cup V_0^j$ intersects the neighborhood $N(s)$.
\end{lemma}

Without revising the definition of $H(G,S)$, the proof of Lemma~\ref{meaningofFptotal} now shows that
\begin{lemma}\label{newmeaningofFptotal}
$H(G,S)$ is the number of covers of $V\setminus S$ by $k$-tuples of sets $I_0,\ldots,I_{k-1}$ such that for every $j\in[k]$, $I_j\cup V_0^j$ is an independent set in $G$ and that for every $s\in S$ the neighborhood $N(s)$ is not intersected by at least one of the $k$ independent sets $\left(I_j\cup V_0^j\right)$ for $j\in[k]$.
\end{lemma}

Therefore, we have
\begin{lemma}
Let $G$ be a graph, $V_0\subseteq V(G)$ a subset of its vertices and $c:V_0 \rightarrow [k]$ a proper $k$-coloring of $G[V_0]$. Denote by $V:=V(G)\setminus V_0$. Let $S\subseteq V$ be an independent set in $G$. Then, $H(G,S)>0$ if and only if $c$ can be extended to a $k$-coloring of $G$.
\end{lemma}

The non-trivial part of the revision and the heart of this subsection, is adjusting the algorithm for computing the values of $h(G,S')$ without increasing the running time.

For $j\in[k]$, denote by 
\[
S_j := \{s \in S\; | \; N(s)\cap V_0^j \neq \emptyset \}
\]
the set of vertices in $S$ whose neighborhood intersects $V_0^j$. 
The key observation of this subsection follows.

\begin{lemma}\label{betarestriction}
For any $j\in[k]$, $S'\subseteq S$, $V' \subseteq V$, we have 
\[
\beta_j \left(V', S'\right) =
\beta_j \left(V', S' \cup S_j\right) 
\]
\end{lemma}
\begin{proof}
For any set $I\subseteq V'$ the set $I \cup V_0^j$ intersects every set in $\{N(s)\}_{s\in S_j}$.
In particular, an independent set $I\subseteq V'$ intersects all of $\{N(s)\}_{s\in S'}$ if and only if it intersects all of $\{N(s)\}_{s\in \left(S' \cup S_j\right)}$.
\end{proof}

Lemma~\ref{betarestriction} implies that it is enough to compute $\beta_j \left(V', S'\right)$ only for sets $S'\subseteq S\setminus S_j$, as its other values can be deduced from these as $\beta_j \left(V', S'\right) = \beta_j \left(V', S'\setminus S_j\right)$.

For any $S'\subseteq S$ we again denote by $B(S') := \{ V'\subseteq V\setminus S\;\; | \;\; \forall s\in S'.\; V'\cap N(s) \neq \emptyset \}$ the set of all subsets of $V\setminus S$ intersecting all neighborhoods of $S'$. Note the slight difference from Section~\ref{maxdegsub} of considering subsets of $V\setminus S$ and not of $V(G)\setminus S$.

As for every $s\in S\setminus S_j$, $N(s)\cap V_0^j = \emptyset$, we still have that for every $V'\notin B(S')$ the value of $\beta_j \left(V', S'\right)$ is zero. In particular, we can still use the implicit Inverse M\"{o}bius Transform of Lemma~\ref{fpcompute} and get

\begin{lemma}\label{newfpcompute}
Assume $S'\subseteq S\setminus S_j$. We can compute $\beta_j \left(V', S'\right)$ for every $V'\in B(S')$ in  $O^*(|B(S')|)$ time. 
\end{lemma}

By Lemma~\ref{sumlemma} we get 
\[
\sum_{S'\subseteq S\setminus S_j} |B(S')| = O^*\left(2^{|V\setminus S|} \cdot (2-2^{-\Delta})^{|S\setminus S_j|}\right).
\]

We can thus compute $\beta_j \left(V', S'\right)$ for every $S'\subseteq S\setminus S_j$ and every $V' \in B(S')$ in $O^*\left(2^{|V\setminus S|} \cdot (2-2^{-\Delta})^{|S\setminus S_j|}\right)$ time. This is the time to emphasise a crucial point. Note that if we would consider every $S'\subseteq S$ instead of $S'\subseteq S\setminus S_j$, then the running time would be $O^*\left(2^{|V\setminus S|} \cdot (2-2^{-\Delta})^{|S\setminus S_j|} \cdot 2^{|S_j|} \right)$, as the neighborhoods corresponding to $S_j$ are intersected by $V_0^j$. This is why we compute every $\beta_j$ separately, and do so for \emph{all} relevant sets $S'$ before computing even a single value $h(G,S')$.
As it always holds that $|S\setminus S_j|\leq |S|$, we conclude that
\begin{corollary}\label{bjs}
We can compute $\beta_j \left(V', S'\right)$ for all $j\in [k]$, $S'\subseteq S\setminus S_j$ and $V' \in B(S')$ in \\ $O^*\left(2^{|V\setminus S|} \cdot (2-2^{-\Delta})^{|S|} \right)$ time.
\end{corollary}
Note that $k=O^*(1)$.
\\
\\
We are now ready to compute the values of $h(G,S')$. 
We start by making the following observation.
\begin{observation}
If $\;\bigcap_{j=0}^{k-1} S_j \neq \emptyset$ then $c$ cannot be extended to a coloring of $G$.
\end{observation}
This holds as if some $s\in S$ has neighbors colored in each of the $k$ colors then it cannot be properly colored. We are thus dealing with the case where $\;\bigcap_{j=0}^{k-1} S_j = \emptyset$.

\begin{lemma}\label{zeroprods}
For any $S'\subseteq S$ and $V'\subseteq V\setminus S$ such that $V'\notin B(S')$ we have
\[
\prod_{j=0}^{k-1} \beta_j \left(V', S'\right) = 0.
\]
\end{lemma}
\begin{proof}
As $V'\notin B(S')$ there exists some $s\in S$ such that $V' \cap N(s) = \emptyset$.
As $\;\bigcap_{j=0}^{k-1} S_j = \emptyset$, there exists a $j\in[k]$ for which $s\notin S_j$. 
Thus, $V_0^j \cap N(s) = \emptyset$ as well. We conclude that $\beta_j \left(V', S'\right) = 0$.
\end{proof}

From Lemma~\ref{betarestriction} and Lemma~\ref{zeroprods} we conclude that
\[
h(G,S') := \sum_{V'\in B(S')} (-1)^{|V|-|V'|} \; \prod_{j=0}^{k-1} \beta_j \left(V', S'\setminus S_j\right).
\]
Thus, we can compute $h(G,S')$ in $O^*(|B(S')|)$ time using the values computed in Corollary~\ref{bjs}.
Using Lemma~\ref{sumlemma} once again, we get that
\[
\sum_{S'\subseteq S} |B(S')| = O^*\left(2^{|V\setminus S|} \cdot (2-2^{-\Delta})^{|S|}\right)
\]
which completes the proof of Theorem~\ref{fixedthm}.
\\
\\
We can now prove Theorem~\ref{mainthmrefined}.
\begin{proof}
We apply the removal lemma of Theorem~\ref{removalthm} to $\mathcal{F}=\{N(s)\}_{s\in S}$ with $C$ to be chosen later.
We thus get a sub-collection $S'\subseteq S$ and a subset of vertices $V_0\subseteq V(G)\setminus S$ such that $|S'| > \rho(\Delta, C)\cdot |S| + C \cdot |V_0|$ and that for every $s_1,s_2 \in S'$ it holds that $N(s_1)\cap N(s_2) \subseteq V_0$.
Denote by $V = V(G)\setminus \left(S' \cup V_0\right)$.
We enumerate over all $k$-colorings $c : V_0 \rightarrow [k]$. 
For each coloring $c$, we check if it is a proper $k$-coloring of $G[V_0]$ and if so we apply Theorem~\ref{fixedthm} on $G$ with $V_0$, $c$, $S'$.
If any of the applications of Theorem~\ref{fixedthm} returned that there exists a valid extension of $c$ to a coloring of $G$, we return that $G$ is $k$-colorable, and otherwise that it is not. 

The running time of the entire algorithm, up to polynomial factors, is 
\[
k^{|V_0|} \cdot \left(2^{|V\setminus S'|} \cdot (2-2^{-\Delta})^{|S'|}\right) =
2^{|V|} \cdot k^{|V_0|} \cdot (1-2^{-\left(\Delta+1\right)})^{|S'|} \leq
2^{|V|} \cdot k^{|V_0|} \cdot (1-2^{-\left(\Delta+1\right)})^{\rho(\Delta, C)\cdot |S| + C \cdot |V_0|}.
\]
By picking $C=\frac{\log k}{-\log \left(1-2^{-\left(\Delta+1\right)}\right)} > 0$ we have
\[
k^{|V_0|} \cdot (1-2^{-\left(\Delta+1\right)})^{C \cdot |V_0|} = 1
\]
and thus the running time is bound by
\[
2^{|V|} \cdot (1-2^{-\left(\Delta+1\right)})^{\rho(\Delta, C)\cdot |S|}.
\]
\end{proof}

\subsection{Generalization to List Coloring}\label{genlstcol}
In this Section we deduce Theorem~\ref{mainthmlst}.
\begin{proof}
Let~$G=(V,E)$ be a~$(\alpha, \Delta)$-bounded graph with color lists~$C_v$ of size at most~$k$ for each~$v\in V$.
Denote by~$U=\bigcup_{v\in V} C_v$ the color universe.
Note that~$|U|$ might be as large as~$kn$, where $n=|V|$.
We construct a new graph~$G'$ on the set of vertices~$V\cup U$ by adding~$|U|$ isolated vertices to the graph~$G$ and then connecting each node~$v\in V$ to every node~$u\in U$ such that~$u\notin C_v$.
If we color each~$u\in U$ by the color~$u$, then there is an extension of this coloring to a (regular) $|U|$-coloring for all of~$G'$ if and only if~$G$ is list-colorable.

We now follow the proof of Theorem~\ref{mainthmrefined}. 
We can again find a subset $S\subset V$ of size $|S|\geq \frac{\alpha}{1+\Delta} n$ which is an independent set in $G$ that contains only vertices of degree at most~$\Delta$.
We then apply the removal lemma of Theorem~\ref{removalthm} to $\mathcal{F}=\{N(s)\}_{s\in S}$ where the neighbourhoods are within~$G$ and~$C$ is to be chosen later.
Define~$S'$ and~$V_0$ as in the proof of Theorem~\ref{mainthmrefined}.
Since every color-list~$C_v$ is of size at most~$k$, there are only~$k^{|V_0|}$ possible colorings~$c:V_0\rightarrow U$.
We enumerate over these colorings and for each one which is a proper coloring of~$G[V_0]$ we apply Theorem~\ref{fixedthm} on~$G'$ where $U\cup V_0$ are already colored (by their corresponding colors and by~$c$, respectively).
The crucial point here is that in Theorem~\ref{fixedthm} the running time depends on~$|V|$, the number of uncolored vertices, and is independent of the number of colored vertices.
In particular, the total runtime is thus 
\[
k^{|V_0|} \cdot \left(2^{|V\setminus \left(S'\cup V_0\right)|} \cdot (2-2^{-\Delta})^{|S'|}\right) \leq
2^{|V|} \cdot k^{|V_0|} \cdot (1-2^{-\left(\Delta+1\right)})^{|S'|} \leq
2^{|V|} \cdot k^{|V_0|} \cdot (1-2^{-\left(\Delta+1\right)})^{\rho(\Delta, C)\cdot |S| + C \cdot |V_0|}.
\]
We can thus again pick $C=\frac{\log k}{-\log \left(1-2^{-\left(\Delta+1\right)}\right)} > 0$ to have
\[
k^{|V_0|} \cdot (1-2^{-\left(\Delta+1\right)})^{C \cdot |V_0|} = 1
\]
which results in a running time bounded by
\[
2^{|V|} \cdot (1-2^{-\left(\Delta+1\right)})^{\rho(\Delta, C)\cdot |S|}.
\]
\end{proof}

\subsection{On finding a coloring}\label{findingcol}
In both previous subsections, we used the bounds on the degrees only in order to construct a \emph{good} independent set $S$.
After doing so, we may apply the self-reduction of Section~\ref{decvssearch} to the graph $G[V(G)\setminus S]$, in which we no longer care about the number of edges nor the degrees.
This would result in finding a $k$-coloring of $G[V(G)\setminus S]$.
Such coloring can be extended to a $k$-coloring of $G$ by the constructive proof of Lemma~\ref{intuitionformalized}. The exact claim follows.

\begin{lemma}
In the conditions of Theorem~\ref{easymainthm} or Theorem~\ref{mainthmrefined} we can also \emph{find} a $k$-coloring of $G$.
\end{lemma}
\begin{proof}
Consider the reduction between the decision and search versions of $k$-coloring of Lemma~\ref{decsearch}.
Since adding edges to vertices whose both endpoints are in $V(G)\setminus S$ does not violate the conditions of the theorems, we may apply the reduction of Lemma~\ref{decsearch} to $G[V(G)\setminus S]$.
By the end of the reduction, we have a $k$-coloring of $G[V(G)\setminus S]$ that is a restriction of some $k$-coloring of $G$. 
We can extend this $k$-coloring to a $k$-coloring of $G$ using the algorithm of Lemma~\ref{intuitionformalized}.
\end{proof}

As a corollary, in the conditions of Theorem~\ref{mainthm} we can also \emph{find} a $k$-coloring of $G$.

\section{Removal Lemma For Small Sets}\label{sec:removal}

In this section we show that any collection of small sets must contain a \emph{large} sub-collection of \emph{almost} pairwise-disjoint sets. The precise statement follows.

\begingroup
\def\thetheorem{\ref{removalthm}}
\begin{theorem}
Let $\mathcal{F}$ be a collection of subsets of a universe $U$ such that every set $F\in \mathcal{F}$ is of size $|F|\leq \Delta$. Let $C>0$ be any constant. Then, there exist subsets $\mathcal{F}'\subseteq \mathcal{F}$ and $U'\subseteq U$, such that
\begin{itemize}
    \item $|\mathcal{F}'| > \rho(\Delta,C)\cdot |\mathcal{F}| + C\cdot |U'|$, where $\rho(\Delta,C)>0$ depends only on $\Delta, C$.
    \item The sets in $\mathcal{F}'$ are disjoint when restricted to $U\setminus U'$, i.e., for every $F_1,F_2 \in \mathcal{F}'$ we have $F_1\cap F_2 \subseteq U'$.
\end{itemize}
\end{theorem}
\addtocounter{theorem}{-1}
\endgroup

We should think of the statement of Theorem~\ref{removalthm} in the following manner. 
We interpret an \emph{almost} pairwise-disjoint sub-collection as a sub-collection that would become pairwise-disjoint after the removal of a \emph{small} number of elements of the universe.
If $\Delta$ is a constant, then the precise meaning of \emph{small} and \emph{large} is that on the one hand, the size of the sub-collection is at least a constant fraction of the size of the entire collection, and on the other hand, its size is arbitrarily larger than the number of removed universe elements.
The constant $C$ represents the exact meaning of \emph{arbitrarily larger}.

\begin{definition}
For $u\in U$, denote by $\deg(u):= |\{F\in \mathcal{F}\;|\;u\in F\}|$ the number of sets in $\mathcal{F}$ containing it.
\end{definition}

We may think of our collection as a bipartite graph, where the left side consists of a vertex for each set in $\mathcal{F}$, the right side consists of a vertex for each element in the universe $U$, and every set $F$ is connected to each universe element it contains. Then, the degree of a universe element $u$ is simply the degree of the vertex corresponding to it in this graph.

We begin by repeating and slightly generalizing Lemma~\ref{maxdeglemma}, focusing on the case in which the universe has bounded degrees.

\begin{lemma}\label{boundeddegs}
Assume that for every $u\in U$ we have $\deg(u) \leq d$, then we can construct a \emph{pair-wise disjoint} sub-collection $\mathcal{F}' \subseteq \mathcal{F}$ of size $\mathcal{F}'\geq \frac{1}{\Delta\cdot d}|\mathcal{F}|$.
\end{lemma}

\begin{proof}
We construct $\mathcal{F}'$ in a greedy manner.
Begin with $\mathcal{F}' = \emptyset$ being the empty set, and $\hat{\mathcal{F}}:=\mathcal{F}$ being the entire collection.
As long as $\hat{\mathcal{F}}$ is non-empty, take an arbitrary element $F$ out of it, add $F$ to $\mathcal{F}'$, and remove every $F'\in\hat{\mathcal{F}}$ intersecting $F$ (including itself) from $\hat{\mathcal{F}}$.
As each set $F$ contains at most $\Delta$ elements and each element is contained in at most $d$ sets, the number of sets intersecting a specific $F$ is bounded by $\Delta \cdot d$.
Thus, the size of $\hat{\mathcal{F}}$ can decrease by at most $\Delta \cdot d$ after every step and therefore we manage to add at least $\frac{1}{\Delta\cdot d}|\mathcal{F}|$ sets to $\mathcal{F}'$ before $\hat{\mathcal{F}}$ becomes empty.
\end{proof}

This leads to a very natural approach. 
As the number of high-degree universe elements should be small, we may try removing all universe elements of degree above some threshold $d$ and then use Lemma~\ref{boundeddegs}.

As the sum of the degrees in each side of a bipartite graph is equal, we have $\sum_{u\in U} \deg(u) = \sum_{F\in\mathcal{F}} |F| \leq \Delta \cdot |\mathcal{F}|$. In particular, the number of universe elements of degree at least $d$ is at most $\frac{\Delta}{d} |\mathcal{F}|$.
Unfortunately, this is larger than $\frac{1}{\Delta\cdot d}|\mathcal{F}|$, the size of the sub-collection we get by Lemma~\ref{boundeddegs}, even if we ignore $C$.

On the other hand, we may notice that the worst-case collection for Lemma~\ref{boundeddegs} is in fact not difficult to deal with.
Consider the case where the non-isolated vertices corresponding to universe elements are regular, i.e., for each $u\in U$ we have $\deg(u)\in\{0,d\}$ for some $d$.
In that case, the number of relevant universe elements is indeed $\frac{\Delta \cdot |\mathcal{F}|}{d}$, but after removing them, the \emph{entire} collection becomes pair-wise disjoint.
Thus, we may either get $|\mathcal{F}'| = \frac{1}{\Delta\cdot d}|\mathcal{F}|$ and $|U'|=0$ from Lemma~\ref{boundeddegs}, or $|\mathcal{F}'|=|\mathcal{F}|$ and $|U'| = \frac{\Delta \cdot |\mathcal{F}|}{d}$ by removing the entire relevant universe. It is easy to verify that for each $C,d$ at least one of the two is large enough, in particular, for every $d$ we get $|\mathcal{F}'|-C\cdot |U'|\geq \frac{1}{1+C\Delta^2} |\mathcal{F}|$ for at least one of them.
Our proof captures this observation.

\begin{definition}
Denote by $U(d):=|\{u\in U\;|\;\deg(u)=d\}|$ the number of universe elements of degree $d$.
\end{definition}

The counting claim regarding the sum of the degrees in each side of the discussed bipartite graph can now be rephrased as
\begin{corollary}\label{udsum}
We have $\sum_{d=1}^{|\mathcal{F}|} d\cdot U(d) \leq \Delta |\mathcal{F}|$.
\end{corollary}

Denote by 
\[
V(d) := \frac{1}{\Delta \cdot d} |\mathcal{F}| - C \sum_{i=d+1}^{|\mathcal{F}|} U(i) 
\]

\begin{lemma}\label{vds}
For any $d\in \mathbb{N}$, we can construct a sub-collection $\mathcal{F}'\subseteq \mathcal{F}$ that is pair-wise disjoint after the removal of a subset $U' \subseteq U$ of the universe, such that $\left(|\mathcal{F}'|-C\cdot |U'|\right) \geq V(d)$.
\end{lemma}
\begin{proof}
We let $U'$ be the set of all universe elements of degree larger than $d$. Thus, $|U'| = \sum_{i=d+1}^{|\mathcal{F}|} U(i)$.
By Lemma~\ref{boundeddegs}, after the removal of $U'$, we can construct a sub-collection of size $|\mathcal{F}'|\geq \frac{1}{\Delta \cdot d} |\mathcal{F}|$. 
\end{proof}

By Lemma~\ref{vds}, in order to prove Theorem~\ref{removalthm} it is enough to give a lower bound for $\max_d V(d)$ that is proportional to $|\mathcal{F}|$. Our approach for this maximization is analytical in nature, but we phrase it in a discrete manner for simplicity.  

We remind the reader of the following well-known fact regarding the Harmonic series
\begin{claim}\label{harmonic}
Let $H_n := \sum_{k=1}^{n} 1/k$ be the Harmonic series. Then, $H_n\geq \ln n + \gamma$, where $\gamma>0.577$ is the Euler-Mascheroni constant. Furthermore, for any $n\geq 1$ it holds that $H_{\lfloor e^n \rfloor} \geq n$.
\end{claim}

We are now ready for the main Lemma needed for the proof of Theorem~\ref{removalthm}.
\begin{lemma}\label{partialsum}
For any $D\in \mathbb{N}$, we have $\sum_{d=1}^{D} V(d) \geq \left(\frac{1}{\Delta}\cdot |\mathcal{F}|\cdot H_D - C\Delta \cdot |\mathcal{F}|\right)$.
\end{lemma}
\begin{proof}
Consider the sum
\begin{align*}
    \sum_{d=1}^{D} V(d) &=
    \sum_{d=1}^{D} \left( \frac{|\mathcal{F}|}{\Delta \cdot d} - C\cdot \sum_{i=d+1}^{|\mathcal{F}|} U(i) \right)  \\ &=
    \frac{|\mathcal{F}|}{\Delta} H_{D} - C\cdot \sum_{d=1}^{ D} \sum_{i=d+1}^{|\mathcal{F}|} U(i) 
\end{align*}
We also notice that
\begin{align*}
\sum_{d=1}^{ D } \sum_{i=d+1}^{|\mathcal{F}|} U(i) &\leq 
\sum_{d=1}^{ |\mathcal{F}|} \sum_{i=d+1}^{|\mathcal{F}|} U(i) \\ &=
\sum_{d=1}^{ |\mathcal{F}|} \sum_{i=1}^{|\mathcal{F}|} \mathbb{1}_{i>d} \cdot U(i) \\ &=
\sum_{i=1}^{|\mathcal{F}|} U(i) \sum_{d=1}^{ |\mathcal{F}|} \mathbb{1}_{i>d} \\ &=
\sum_{d=1}^{ |\mathcal{F}|} (i-1)\cdot U(i) \\ & \leq
\sum_{d=1}^{ |\mathcal{F}|} i\cdot U(i) \leq \Delta |\mathcal{F}|
\end{align*}
where the last inequality follows from Corollary~\ref{udsum} and the indicator $\mathbb{1}_{i>d}$ is defined to be $1$ if $i>d$ and $0$ otherwise. Combining both inequalities we have
\[
\sum_{d=1}^{D} V(d) 
= \frac{|\mathcal{F}|}{\Delta} H_{D} - C\cdot \sum_{d=1}^{ D} \sum_{i=d+1}^{|\mathcal{F}|} U(i)
\geq \frac{|\mathcal{F}|}{\Delta} H_{D} - C\Delta \cdot |\mathcal{F}|
\]
proving the Lemma's statement.
\end{proof}

By a standard averaging argument we have
\begin{observation}\label{partsumobs}
For any $D\in \mathbb{N}$, $\max_d V(d) \geq \frac{1}{D} \sum_{d=1}^{D} V(d)$.
\end{observation}

We are now ready to finish the proof of Theorem~\ref{removalthm}.
\begin{proof}
By Observation~\ref{partsumobs} and Lemma~\ref{partialsum}, we have that for every $D\in \mathbb{N}$,
\[
\max_d V(d) \geq \frac{1}{D} \left(\frac{|\mathcal{F}|}{\Delta} H_{D} - C\Delta \cdot |\mathcal{F}|\right).
\]

Slightly rearranging the inequality gives
\[
\frac{1}{|\mathcal{F}|}\max_d V(d) \geq \frac{1}{D} \left(\frac{1}{\Delta} H_D - C\Delta\right)
\]
then choosing an optimal $D=\lfloor exp(1+C\Delta^2)\rfloor$ and using Claim~\ref{harmonic} yields
\[
\frac{1}{|\mathcal{F}|}\max_d V(d) 
\geq \frac{1}{exp(1+C\Delta^2)} \left(\frac{1}{\Delta} \left(1+C\Delta^2\right) - C\Delta\right)
> \frac{1}{\Delta\cdot e^{1+C\Delta^2}}.
\]
We complete the proof by using Lemma~\ref{vds}.
\end{proof}

We also observe that the sets $\mathcal{F}',U'$ satisfying the statement of Theorem~\ref{removalthm} can be computed efficiently.
\begin{observation}\label{computeremoval}
Assuming $\Delta, C$ are constants, sets $\mathcal{F}',U'$ satisfying the statement of Theorem~\ref{removalthm} can be computed in time $O(|\mathcal{F}|)$.
\end{observation}
\begin{proof}
Computing the bipartite graph representing the collection $\mathcal{F}$ takes linear time.
Then, the proofs of Lemma~\ref{boundeddegs} and Lemma~\ref{vds} are both constructive and run in linear time for any specific degree threshold $d$.
By the proof of Theorem~\ref{removalthm} we notice that it is enough to consider only degree thresholds in the constant-sized range $1\leq d \leq D=\lfloor exp(1+C\Delta^2)\rfloor = O(1)$. 
Thus, we may find a threshold satisfying the statement in linear time.
\end{proof}

In Appendix~\ref{removallb} we present a construction showing that in the settings of Theorem~\ref{removalthm} we must have $\rho(\Delta,C)\leq (C+1)^{-\Delta}$. 
The construction is due to discussions with Noga Alon.

\section{Reducing $k$-coloring to $(k-1)$-list-coloring}\label{sec:red}
In this section we use Theorem~\ref{mainthm} in order to prove the following.

\begingroup
\def\thetheorem{\ref{reduct}}
\begin{theorem}
Given an algorithm solving $(k-1)$-list-coloring in time $O\left(\left(2-\varepsilon\right)^n\right)$ for some constant $\varepsilon>0$, we can construct an algorithm solving $k$-coloring in time $O\left(\left(2-\varepsilon'\right)^n\right)$ for some (other) constant $\varepsilon'>0$.
Furthermore, the reduction is deterministic.
\end{theorem}
\addtocounter{theorem}{-1}
\endgroup

Beigel and Eppstein \cite{beigel20053} show that $4$-list-coloring (as a special case of a $(4,2)$-CSP) can be solved in time $O(1.81^n)$. Therefore we conclude that 

\begingroup
\def\thetheorem{\ref{5col}}
\begin{theorem}
$5$-coloring can be solved in time $O\left(\left(2-\varepsilon\right)^n\right)$ for some constant $\varepsilon>0$.
\end{theorem}
\addtocounter{theorem}{-1}
\endgroup

We begin by illustrating the idea intuitively.
By Theorem~\ref{mainthm}, it suffices to solve $k$-coloring for graphs in which most vertices have high degrees.
We show that in this case, the graph has a small \emph{dominating set}, this is a subset $R$ of vertices such that every vertex not in $R$ is adjacent to at least one vertex of $R$.
Given a $k$-coloring of the dominating set, the problem of extending the coloring to a $k$-coloring of the entire graph becomes a problem of $(k-1)$-list-coloring the rest of the graph. This is because each vertex not in the dominating set has a neighbor in it, and thus has at least one of the $k$ colors which it cannot use.
Assuming the dominating set is small enough, we can enumerate over the $k$-colorings of vertices in it, and then solve the remaining $(k-1)$-list-coloring problem.

\begin{lemma}\label{domset}
Let $G$ be a graph. Assume that there exists a subset of vertices $V'\subseteq V(G)$ of size $|V'|\geq (1-\alpha)\cdot |V(G)|$ such that for every $v\in V'$ we have $\deg(v)\geq \Delta-1$. Then, $G$ has a dominating set $R\subseteq V(G)$ of size $|R|\leq \left(\left(1-\alpha\right)\cdot \frac{1 + \ln \Delta}{\Delta} +  \alpha \right)\cdot |V(G)|$. Furthermore, there is an efficient deterministic algorithm to find such a dominating set.
\end{lemma}
\begin{proof}
Denote by $\delta(G)$ the minimum degree of a vertex in $G$.
Let $R_0$ be a random subset of $V(G)$ chosen by picking each $v\in V(G)$ independently with probability $p$. We have $E[|R_0|]=p\cdot|V(G)|$.
Let $v\in V'$ be a vertex of degree at least $\Delta-1$, the probability of not adding $v$ or any one of its neighbors to $R$ is at most $(1-p)^\Delta$.
Denote by $R_1$ the set of vertices that are in $V'$ but not in $R_0$ and do not have a neighbor in $R_0$. 
By the previous observation, $E[|R_1|]\leq (1-p)^\Delta \cdot |V'|$.
Similarly, denote by $R_2$ the set of vertices that are in $V(G)\setminus V'$, not in $R_0$ and do not have a neighbor in $R_0$. 
We have $E[|R_2|]\leq (1-p)^{\delta(G)+1} \cdot |V(G)\setminus V'|$.

The set $R=R_0\cup R_1 \cup R_2$ is a dominating set. We have
\begin{align*}
E[|R|] &\leq E[|R_0|]+E[|R_1|]+E[|R_2|] \\& \leq 
p\cdot|V(G)| + (1-p)^\Delta \cdot |V'| + (1-p)^{\delta(G)+1} \cdot |V(G)\setminus V'| \\ &\leq
\left(p + \left(1-\alpha\right)\cdot (1-p)^\Delta + \alpha \cdot (1-p)^{\delta(G)+1} \right) \cdot |V(G)|.
\end{align*}
Furthermore, we can efficiently compute $E[|R|]$ even after conditioning on whether or not vertices are chosen to $R_0$. Thus, the \emph{method of conditional expectations} results in an efficient deterministic algorithm that finds a set $R$ of size at most the above expectation. We elaborate on the matter in Appendix~\ref{derand}.

While not optimal for many parameters, for the sake of our use of this lemma it suffices to pick $p=\frac{\ln \Delta}{\Delta}$, for which we get
\begin{align*}
|R| &\leq 
\left(\frac{\ln \Delta}{\Delta} + \left(1-\alpha\right)\cdot (1-\frac{\ln \Delta}{\Delta})^\Delta + \alpha \cdot (1-\frac{\ln \Delta}{\Delta})^{\delta(G)+1} \right) \cdot |V(G)| \\ &\leq
\left(\frac{\ln \Delta}{\Delta} + \left(1-\alpha\right)\cdot e^{-\ln \Delta} + \alpha \cdot (1-\frac{\ln \Delta}{\Delta}) \right) \cdot |V(G)| \\ &=
\left(\left(1-\alpha\right)\cdot \frac{1 + \ln \Delta}{\Delta} +  \alpha \right) \cdot |V(G)|.
\end{align*}

\end{proof}

As evident in the proof of Lemma~\ref{domset}, a lower bound on the minimum degree $\delta(G)$ of the graph can result in a slightly better bound on the size of the dominating set we can construct. While it is not necessary for proving the statement of Theorem~\ref{reduct}, we include the following observation for completeness.

\begin{lemma}\label{removelowdeg}
Given an algorithm solving $k$-coloring for graphs of minimum degree $\delta(G)\geq k$, we can construct an algorithm solving $k$-coloring for every graph with the same running time (up to an additive polynomial factor).
\end{lemma}
\begin{proof}
Denote by $\mathcal{A}$ the algorithm solving $k$-coloring for graphs with minimum degree $\delta(G)\geq k$. 
Given a graph $G$, we initiate a stack $\sigma$ and run the following iterative process.
As long as there is a vertex $v$ in $G$ of degree $\deg(v) < k$, we push $v$ into $\sigma$ and remove it and its adjacent edges from $G$.
When we finish, our graph is of minimal degree $\delta(G)\geq k$ and thus we can run $\mathcal{A}$.
If $G$, which is currently an induced sub-graph of the input graph, is not $k$-colorable, then the input graph is not $k$-colorable as well.
Otherwise, we extend the coloring $c$ of $G$ returned by $\mathcal{A}$ iteratively as follows.
As long as $\sigma$ is not empty, pop a vertex $v$ out of it. 
Re-insert $v$ and its adjacent edges back into $G$. 
As by construction it is of degree $\deg(v)<k$, we must have at least one color $i$ that is not used for any of $v$'s neighbors. 
Extend $c$ to $v$ by setting $c(v)=i$.
When the stack $\sigma$ is empty, $c$ is a $k$-coloring of the entire input graph.


\end{proof}

\begin{lemma}\label{usedomset}
Let $G$ be a graph with a dominating set $R$. We can solve $k$-coloring for $G$ by solving $k^{|R|}$ instances of $(k-1)$-list-coloring on graphs with $|V(G)|-|R|$ vertices.
\end{lemma}
\begin{proof}
A $k$-coloring $c:R\rightarrow [k]$ of $G[R]$ can be \emph{extended} to a $k$-coloring $c':V(G)\rightarrow [k]$ of $G$ with $c'|_{R}=c$, if and only if there is valid coloring of $G[V(G)\setminus R]$ such that a vertex $v\in V(G)\setminus R$ can only be colored with a color from $[k]\setminus c(R\cap N(v))$. As each $v\in V(G)\setminus R$ has at least one neighbor in $R$, we have $|R\cap N(v)|\geq 1$ and in particular $|[k]\setminus c(R\cap N(v))|\leq k-1$. Thus, we are left with a $(k-1)$-list-coloring problem on $G[V(G)\setminus R]$.
\end{proof}

We are now ready to prove Theorem~\ref{reduct}.
\begin{proof}
Let $\Delta, \alpha>0$ be constants to be chosen later.
Given a graph $G$ with $n$ vertices, we check whether it is \mbox{$(\alpha,\Delta)$-bounded}. If it is, then we use the algorithm of Theorem~\ref{mainthm} to solve $k$-coloring in $O\left(\left(2-\varepsilon_{\alpha,\Delta}\right)^n \right)$ time.
Otherwise, there are more than $(1-\alpha)n$ vertices of degree larger than $\Delta$ and by Lemma~\ref{domset} we can find a dominating set $R$ of $G$ of size $|R|\leq \left(\left(1-\alpha\right)\cdot \frac{1 + \ln \Delta}{\Delta} +  \alpha \right)n$.
Using Lemma~\ref{usedomset} and the given $(k-1)$-list-coloring algorithm, we can solve $k$-coloring for $G$ in time
\[
k^{|R|} \cdot \left(2-\varepsilon\right)^{n-|R|} = 
\left(\frac{k}{2-\varepsilon}\right)^{|R|} \cdot \left(2-\varepsilon\right)^{n} \leq
\left(\frac{k}{2-\varepsilon}\right)^{\left(\left(1-\alpha\right)\cdot \frac{1 + \ln \Delta}{\Delta} +  \alpha \right)n} \cdot \left(2-\varepsilon\right)^{n}.
\]

Combining both cases, we get an algorithm running in time $O\left(2-\varepsilon'\right)^{n}$ for
\[
\varepsilon' := \min\left(
\varepsilon_{\alpha,\Delta}
\;,\;\;
2 - \left(\frac{k}{2-\varepsilon}\right)^{\left(\left(1-\alpha\right)\cdot \frac{1 + \ln \Delta}{\Delta} +  \alpha \right)} \cdot \left(2-\varepsilon\right)
\right).
\]

When $\alpha\rightarrow 0$ and $\Delta \rightarrow \infty$ the second expression converges to $\varepsilon>0$. Therefore, for any choice of a small enough constant $\alpha$ and large enough integer $\Delta$ we have $\varepsilon'>0$.
\end{proof}

\section{Reducing $k$-coloring to $(k-2)$-list-coloring}\label{sec:red2}
We now refine the reduction of Section~\ref{sec:red} and show that

\begingroup
\def\thetheorem{\ref{reduct2}}
\begin{theorem}
Given an algorithm solving $(k-2)$-list-coloring in time $O\left(\left(2-\varepsilon\right)^n\right)$ for some constant $\varepsilon>0$, we can construct an algorithm solving $k$-coloring with high probability in time $O\left(\left(2-\varepsilon'\right)^n\right)$ for some (other) constant $\varepsilon'>0$.
\end{theorem}
\addtocounter{theorem}{-1}
\endgroup

Once again, we use the $4$-list-coloring algorithm of Beigel and Eppstein \cite{beigel20053} to conclude

\begingroup
\def\thetheorem{\ref{6col}}
\begin{theorem}
$6$-coloring can be solved with high probability in time $O\left(\left(2-\varepsilon\right)^n\right)$ for some constant $\varepsilon>0$.
\end{theorem}
\addtocounter{theorem}{-1}
\endgroup

We begin by outlining the way in which the previous reduction can be improved. 
Consider the reduction of Section~\ref{sec:red} and specifically the proof of Theorem~\ref{reduct}.
In the case where the graph is \mbox{$(\alpha,\Delta)$-bounded}, we may still use Theorem~\ref{mainthm} and gain an exponential improvement.
We now focus on the other case, in which most vertices are of degrees larger than $\Delta$.
Let $\Delta'$ be some constant to be chosen later. 
We think of $\Delta'$ as large yet arbitrarily smaller than $\Delta$.
Consider an arbitrary $k$-coloring $c:V(G)\rightarrow [k]$ of $G$.
For a vertex $v\in V(G)$ we denote by $N_i(v) := N(v)\cap c^{-1}(i)$ the set of $v$'s neighbors that are colored by $i$ in $c$.
We say that a vertex $v$ is \emph{good} if there are at least two distinct colors $i\neq j$ for which $|N_i(v)|,|N_j(v)| > \Delta'$.
As in the proof of Lemma~\ref{domset}, a small random subset of vertices (whose size depends on $\Delta'$) is likely to hit at least one neighbor of $v$ of color $i$ and at least one neighbor of $v$ of color $j$.
Denote by $\beta$ the fraction of \emph{bad} (i.e., not \emph{good}) vertices in $V(G)$.
If $\beta$ is small enough, a reduction almost identical to the previous one works. 
Uniformally pick a random small set $R_0$ of graph vertices, enumerate over the colorings of the vertices in $R_0$. In one of the colorings (the one corresponding to $c$ restricted to $R_0$) we expect having in $R_0$ neighbors of at least two different colors for almost all vertices of $V(G)\setminus R_0$.
With a cautious implementation, this gives a reduction to $(k-2)$-list-coloring.
Thus, the interesting case is when $\alpha$ is very small yet $\beta$ is large.
For a bad vertex $v\in V(G)$ of degree larger than $\Delta$, we must have single color $i$ such that $|N_i(v)| \geq \left(1-\frac{\left(k-1\right)\cdot \Delta'}{\Delta}\right) |N(v)|$.
Thus, almost all of the neighbors of a \emph{bad} vertex can be colored by the same color.
We therefore aim to gain by picking a large subset of $v$'s neighbors and \emph{contract} them to a single vertex. 
It is likely that $c$ remains a valid coloring after the contraction.
Furthermore, if the contracted set is an independent set, a coloring of the resulting graph is also a coloring of the original graph.
The algorithmic harnessing of the above observation is somewhat involved, as we cannot identify \emph{good} and \emph{bad} vertices easily.

\begin{lemma}\label{red2smallb}
Let $G$ be a $k$-colorable graph, $\Delta'$ be some constant, let $c$ be a $k$-coloring of $G$.
Assume that we are also given $\beta$, an upper bound on the fraction of \emph{bad} vertices in $G$ with respect to $c,\Delta'$.
Given an algorithm $\mathcal{A}$ solving $(k-2)$-list-coloring in time $O\left(\left(2-\varepsilon\right)^n\right)$, 
we can construct an algorithm $\mathcal{A}_1$ that runs in 
\[
O\left( 
\left(\frac{k}{2-\varepsilon}\right)^{\left(\frac{6 + \ln \Delta'}{\Delta'} +  \beta \right)n} \cdot \left(2-\varepsilon\right)^{n}
\right)
\]
time, and returns a $k$-coloring of $G$ with probability at least $\frac{1}{2}$.
\end{lemma}
\begin{proof}
Let $R_0$ be a random subset of $G$'s vertices, picking each vertex independently with probability $p$.
Let $v$ be a \emph{good} vertex and $i,j$ two colors for which $|N_i(v)|,|N_j(v)|>\Delta'$. 
We have a probability of at most $2\cdot (1-p)^{\Delta'}$ that either $N_i(V)\cap R_0$ or $N_j(v)\cap R_0$ is empty.
Thus, the expected number of \emph{good} vertices without neighbors in $R_0$ of two different colors (according to $c$) is bounded by $2 (1-p)^{\Delta'} n$.
By Markov's inequality, with probability greater than $\frac{3}{5}$ their number is at most $5 (1-p)^{\Delta'} n$.
For any $R_0 \subseteq V(G)$ and a partial coloring $c' : R_0\rightarrow [k]$, denote by $B(R_0,c')$ the set of all vertices in $V(G)\setminus R_0$ that do not have neighbors of two different colors (according to $c'$) in $R_0$.
By the above, we have that 
\[
| B(R_0, \restr{c}{R_0}) | \leq \beta n + 5 (1-p)^{\Delta'} n
\]
with probability at least $\frac{3}{5}$.
We pick $p = \frac{\ln \Delta'}{\Delta'}$ and have
\[
\beta n + 5\cdot (1-p)^{\Delta'} \cdot n <
\beta n + \frac{5}{\Delta'} n.
\]
We also note that $|R_0| \sim Bin(n,p) = Bin(n,\frac{\ln \Delta'}{\Delta'})$. 
Thus by applying the standard Chernoff bound \mbox{$\Pr\left(X>(1+\delta)\mu \right)< e^{-\frac{\delta^2\mu}{3}}$} with $\delta=\frac{1}{\ln \Delta'}$ we have
\[
Pr(|R_0|>\frac{1+\ln \Delta'}{\Delta'} n) < e^{-\frac{n}{3\Delta' \ln \Delta'}} < \frac{1}{10}.
\]

We therefore consider Algorithm~1.

\begin{figure}[t]
\begin{algorithm}[H]
\caption{Algorithm $\mathcal{A}_1(G,k,\Delta',\beta)$}
Pick $R_0$, a random subset of $V$ where vertices are picked i.i.d. with probability $\frac{\ln \Delta'}{\Delta'}$ \;
\If{$|R_0|>\frac{1+\ln \Delta'}{\Delta'} n$}
{Return that no coloring was found and halt\;}
\For{Every function $c' : R_0 \rightarrow [k]$}
{
    \If{$c'$ is a valid coloring of $G[R_0]$}
    {
        Compute $R:=B(R_0, c')$\;
        \If{$|R| < \beta n + \frac{5}{\Delta'} n$}
        {
            \For{Every function $c'' : R \rightarrow [k]$}
            {
                \If{$c'\cup c''$ is a valid coloring of $G[R_0\cup R]$}
                {
                    \For{$v\in V(G)\setminus (R_0\cup R)$}
                    {
                    	$L(v):=[k]\setminus(c'\cup c'')\left(N(v)\cap(R_0\cup R)\right)$\;
                    }
                    Run $\mathcal{A}$ on $V(G)\setminus (R_0\cup R)$ with the lists $L(\cdot)$\;
                    If it returns a coloring $c'''$, return $c'\cup c'' \cup c'''$ and halt\;
                }
            }
        }
    }
}
Return that no coloring was found\;
\end{algorithm}
\end{figure}

The correctness is quite straightforward.
Every coloring returned by the algorithm is valid, and with probability at least $\frac{1}{2}$ we reach the inner for loop with both $c' \cup c'' = \restr{c}{R_0\cup R}$ and thus $\mathcal{A}$ will return a valid solution.
The inner loops run at most $k^{\frac{1+\ln \Delta'}{\Delta'} n + \beta n + \frac{5}{\Delta'} n}$ times and thus we get the desired running time.
\end{proof}

If we choose a large enough constant $\Delta'$ and $\beta$ is small enough, Lemma~\ref{red2smallb} gives an exponential improvement.
We next deal with the case where $\beta$ is not small enough, and then finally discuss our concrete algorithm (that cannot compute or use the value of $\beta$).

\begin{lemma}\label{red2largeb}
Let $G$ be a $k$-colorable graph, $\Delta', r\geq 2$ be some integers, and let $c$ be a $k$-coloring of $G$.
Let $\beta$ be a lower bound on the fraction of \emph{bad} vertices in $G$ with respect to $c,\Delta'$.
Denote by $\Delta := r (k-1) \Delta' + r^2$ and by $\alpha$ the fraction of $G$'s vertices of degrees at most $\Delta$.

If we pick a random vertex $v\in V(G)$ and then a random subset $S\subseteq N(v)$ of size exactly $r$, then the probability that $c(u)$ is identical for all $u\in S$ is at least $\frac{1}{4}(\beta-\alpha)$. 
\end{lemma}
\begin{proof}
With probability at least $\beta-\alpha$ the vertex $v$ is \emph{bad} and of degree larger than $\Delta$. 
In this case, there exists a single color $i$ such that for all $j\neq i$ we have $|N_j(v)|<\Delta'$.
We construct $S$ iteratively by picking a random neighbor of $v$ that is not already in $S$ for $r$ times.
After $\ell<r$ iterations, the probability of a random vertex of $N(v)\setminus S$ to be in $N_i(v)$ is at least \[
\frac{|N_i(v)| - r}{|N(v)|} = 1 - \frac{r + \sum_{j\neq i} |N_j(v)|}{|N(v)|} 
\geq
1 - \frac{r + (k-1)\cdot \Delta'}{\Delta} =
1 - \frac{1}{r}.
\]
Thus, the probability that all $r$ neighbors are in $N_i(v)$ is at least 
\[
\left(1-\frac{1}{r}\right)^r \geq \frac{1}{4}.
\]
\end{proof}

Intuitively, if $\frac{1}{4}(\beta-\alpha) > 2^{-(r-1)}$ it is beneficial to use Lemma~\ref{red2largeb} and contract the set $S$, decreasing the number of vertices by $(r-1)$.

For constants $\Delta', r$ we set 
\begin{align*}
    \Delta &:= r^2 + r(k-1)\Delta' \\
    \beta' &:= 8 \cdot \frac{(2-\varepsilon )^{-(r-1)}}{1 - \left(2-\varepsilon\right)^{-r}}\\
    \alpha' &:= \frac{1}{2} \beta'.
\end{align*}
Furthermore, we pick $\Delta',r$ to be large enough to satisfy
\[
\left(\frac{k}{2-\varepsilon}\right)^{\left(\frac{6 + \ln \Delta'}{\Delta'} +  \beta' \right)} \cdot \left(2-\varepsilon\right) < 2
.
\]

Let $\varepsilon''$ be $\left(2 - \left(\frac{k}{2-\varepsilon}\right)^{\left(\frac{6 + \ln \Delta'}{\Delta'} +  \beta' \right)} \cdot \left(2-\varepsilon\right)\right) \in (0,\varepsilon)$.
Let $\varepsilon'$ the minimum between $\varepsilon_{k,\Delta,\alpha'}$ of Theorem~\ref{mainthm} and $\varepsilon''$.

\begin{figure}[t]
\begin{algorithm}[H]
\caption{Algorithm $\mathcal{A}_2(G,k,\Delta',\Delta,\alpha',\beta')$}
\eIf{$G$ is \mbox{$(\alpha',\Delta)$-bounded}}
{
   Run Algorithm $\mathcal{A}_{k,\Delta,\alpha'}(G)$\;
}
{
    Run Algorithm $\mathcal{A}_1(G,k,\Delta',\beta')$\;
}
\end{algorithm}
\end{figure}

We first combine Algorithm $\mathcal{A}_{k,\Delta,\alpha'}$ of Theorem~\ref{mainthm} and Algorithm $\mathcal{A}_1$ of Lemma~\ref{red2smallb} and define Algorithm~2 that covers both the case when $\beta$ is small and the case when $\alpha$ is large.

The following Lemma immediately follows 
\begin{lemma}\label{algapp}
Algorithm $\mathcal{A}_2$ runs in $O\left(\left(2-\varepsilon'\right)^n\right)$ time, and if $(\beta-\alpha) \leq (\beta'-\alpha')$ it returns a $k$-coloring of $G$ with probability at least $\frac{1}{2}$.
\end{lemma}
\begin{proof}
If $\alpha \geq \alpha'$ then we run Algorithm $\mathcal{A}_{k,\Delta,\alpha'}$ and thus correctness follows from Theorem~\ref{mainthm}.
Otherwise, $\beta \leq \alpha + (\beta'-\alpha') < \beta'$ and thus correctness follows from Lemma~\ref{red2smallb}.
\end{proof}

We finally prove Theorem~\ref{reduct2} by constructing Algorithm~3.

\begin{figure}[t]
\begin{algorithm}[H]
\caption{Algorithm $\mathcal{A}_3(G,k)$}
    Set \emph{flag} to $1$ with probability $\left(2-\varepsilon\right)^{-|V(G)|}$, and to $0$ otherwise\;
    \eIf{\emph{flag} is $1$ or $|V(G)|\leq r$}
    {
        Run Algorithm $\mathcal{A}_2(G,k,\Delta',\Delta,\alpha',\beta')$ and return its output\;
    }
    {
       Choose a random $v$ uniformly out of $V(G)$\;
       \If{$\deg(v) < \Delta$}{Halt\;}
       Choose uniformally a random subset $S\subseteq N(v)$ of size exactly $r$\;
       \If{$S$ is not an independent set in $G$}{Halt\;}
       Contract $S$ to a single vertex in $G$\;
       Run $\mathcal{A}_3(G)$ recursively. If the recursive call returned a coloring, we convert it to a coloring of the original graph by expanding the contracted vertex back into $S$ and giving all of its vertices the color of the contracted vertex\;
    }
\end{algorithm}
\end{figure}

\begin{lemma}
If $G$ is $k$-colorable then $\mathcal{A}_3(G)$ returns a coloring with probability\\ at least $\left(2-\varepsilon\right)^{-(|V(G)|+1)}$.
\end{lemma}
\begin{proof}
We prove the claim by induction on $|V(G)|$.
The base case $|V(G)|\leq r$ follows from the correctness of Algorithm $\mathcal{A}_2$. 
We now prove the induction step by considering two cases.
If $(\beta-\alpha) \leq (\beta'-\alpha')$ then with probability $\left(2-\varepsilon\right)^{-|V(G)|}$ we set \emph{flag} to $1$ and run Algorithm $\mathcal{A}_2$. We then produce a coloring with probability at least $\frac{1}{2}$ by Lemma~\ref{algapp}.
Otherwise, with probability $1-\left(2-\varepsilon\right)^{-|V(G)|} \geq 1-\left(2-\varepsilon\right)^{-r}$ we set \emph{flag} to $0$.
We have
\[
\frac{1}{4}(\beta-\alpha) > 
\frac{1}{4}(\beta' - \alpha') =
\frac{1}{8} \beta' =
\frac{(2-\varepsilon )^{-(r-1)}}{1 - \left(2-\varepsilon\right)^{-r}},
\]
and thus by Lemma~\ref{red2largeb} we both set \emph{flag} to $0$ and pick a set $S$ such that $G$ remains $k$-colorable after the contraction with probability greater than $(2-\varepsilon )^{-(r-1)}$.
The contraction decreases $|V(G)|$ by exactly $r-1$ and therefore by the induction hypothesis the probability of the recursive call to return a coloring is at least $(2-\varepsilon )^{-(|V(G)|-(r-1)+1)}$.
In both cases, the induction hypothesis holds for $|V(G)|$.
\end{proof}

\begin{lemma}
The expected running time of Algorithm $\mathcal{A}_3$ is $O\left(\left(\frac{2-\varepsilon'}{2-\varepsilon}\right)^{n}\right)$.
\end{lemma}
\begin{proof}
With probability $\left(2-\varepsilon\right)^{-|V(G)|}$ we set \emph{flag} to $1$ and run Algorithm $\mathcal{A}_2$ which takes $O\left(\left(2-\varepsilon'\right)^{|V(G)|}\right)$ time.
Otherwise, we recursively run $\mathcal{A}_3$ on a graph with $|V(G)|-(r-1)$ vertices.
Thus, the expected running time is
\begin{align*}
T(n) &= \left(2-\varepsilon\right)^{-n} \cdot O\left(\left(2-\varepsilon'\right)^{n}\right)
+
\left(1-\left(2-\varepsilon\right)^{-n}\right) \cdot T\left(n-\left(r-1\right)\right)
\\&\leq
O\left(\left(\frac{2-\varepsilon'}{2-\varepsilon}\right)^{n}\right) + T\left(n-\left(r-1\right)\right)
\\ &=\ldots=
O\left(\left(\frac{2-\varepsilon'}{2-\varepsilon}\right)^{n}\right).
\end{align*}
The last equality holds as $\varepsilon' < \varepsilon$.
\end{proof}

\begin{proof}[Proof of Theorem~\ref{reduct2}]
We run $\mathcal{A}_3(G)$ for $n\cdot \left(2-\varepsilon\right)^{n+1}$ times. If any of them found a coloring we return it and otherwise say that the graph is not $k$-colorable.
If $G$ is $k$-colorable, the probability we never find a coloring is bounded by
\[
\left(1-\left(2-\varepsilon\right)^{-(n+1)}\right)^{n\cdot \left(2-\varepsilon\right)^{n+1}}
<
e^{-n}
.
\]
The expected running time of all iterations together is $O\left(n\left(2-\varepsilon'\right)^n\right)$.
We can terminate the run of the algorithm if it takes much longer than its expected run-time as with high probability it does not happen. 
\end{proof}


\section{Conclusions and Open Problems}\label{conclusions}
The main algorithmic contribution of the paper is Theorem~\ref{mainthm}.
We use it in order to answer a few fundamental questions regarding the running time of $k$-coloring algorithms. 
In particular, we present the first $O\left(\left(2-\varepsilon\right)^n\right)$ algorithms solving $5$-coloring and $6$-coloring, for some $\varepsilon>0$.
While the $\varepsilon$ we can get using our tools is very small, this serves as the first proof that $5$-coloring can be solved faster than we can currently compute the chromatic number in general. The upper bound in Appendix~\ref{removallb} shows that the magnitude of $\varepsilon$ is a necessary consequence of using the removal lemma. 

The main open problem that we leave unsettled is 
\begin{open}
Can we solve $k$-coloring in $O^*\left(\left(2-\varepsilon_k \right)^n\right)$ time for some $\varepsilon_k>0$, for every $k$?
\end{open}

Theorem~\ref{mainthm} makes some progress towards answering it, by giving some additional conditions on the input graph under which the answer is affirmative.
In particular, we show that it holds for every graph that does not contain almost only vertices of super-constant degrees.
In \cite{golovnev2016families} very different techniques (using modifications of the FFT algorithm) were used to get a statement similar to Theorem~\ref{mainthm} for graphs with bounded average degree. 
It seems like their methods do not extend to the case of $(\alpha,\Delta)$-bounded graphs, nevertheless, it is intriguing to find out whether a combination of their techniques with these presented in this paper can lead to further improvements.

While it is believed that $O^*(2^n)$ is the right bound for computing the chromatic number, we have no strong evidence to support this. 
There are reductions from popular problems and conjectures (like SETH) to other partitioning problems \cite{cygan2016problems} or other parameterizations of the coloring problem \cite{jaffke2017fine}. 
It is interesting whether it can be showed that an $O^*\left(\left(2-\varepsilon\right)^n\right)$ algorithm for computing the chromatic number would refute any other popular conjecture.
This question was raised several times, including in the book of Fomin and Kratsch \cite{fomin2010exact}.

Another technical contribution of the paper is Theorem~\ref{removalthm}.
We believe that the presented removal lemma could serve as a tool in the design of other exponential time algorithms. 
It would be interesting to find more problems for which it can be used.

\section*{Acknowledgements}
The author would like to deeply thank Noga Alon for important discussions and insights regarding the subset removal lemma, and Haim Kaplan and Uri Zwick for many helpful discussions and comments on the paper.
The author would also like to thank anonymous reviewers for helpful comments.

\bibliographystyle{plain}
\bibliography{references}

\appendix
\section{Appendix}
\subsection{Upper Bound for Section~\ref{sec:removal}}\label{removallb}
In this section we provide a construction showing that in Theorem~\ref{removalthm} we must have $\rho(\Delta,C) \leq (C+1)^{-\Delta}$.
This bound is due to discussions with Noga Alon.

\begin{theorem}\label{removalupperbound}
For any positive \emph{integers} $C,\Delta,n$ we can construct a collection $\mathcal{F}$ of $(C+1)^{\Delta}\cdot n$ sets of size $\Delta$, such that for every subsets $\mathcal{F}'\subseteq \mathcal{F}$ and $U'\subseteq U$ satisfying that $\forall F_1,F_2\in \mathcal{F}'.\;F_1\cap F_2\subseteq U'$, we have $|\mathcal{F}'|-C\dot |U'|\geq n$.
\end{theorem}

Denote by $T$ the complete $(C+1)$-ary tree of depth $\Delta-1$.
Let the universe $U_0$ be the set of $T$'s vertices.
Let $\mathcal{F}_0$ be the collection of $(C+1)\dot (C+1)^{\Delta-1}$ sets corresponding to root-to-leaf paths in $T$ taken with multiplicity $(C+1)$ each. Each set of $\mathcal{F}_0$ contains the $\Delta$ vertices in its corresponding path, each such path has $(C+1)$ identical sets corresponding to it in $\mathcal{F}_0$.

\begin{lemma}
For every $\mathcal{F}'_0\subseteq \mathcal{F}_0$ and $U'_0 \subseteq U_0$ satisfying that $\forall F_1,F_2\in \mathcal{F}'_0.\;F_1\cap F_2\subseteq U'_0$, we have $|\mathcal{F}'_0| - C\cdot |U'_0| \leq 1$.
\end{lemma}
\begin{proof}
Denote by $r$ the root of $T$.
If $r\notin U'_0$ then $|\mathcal{F}'_0|\leq 1$ as all sets contain $r$.
Otherwise, denote by $T_0$ the connected component of $T[U'_0]$ (i.e., the induced sub-graph of $T$ on the vertex set $U'_0$) containing $r$.
Denote by $\ell$ the number of leaves in $T_0$, and by $|T_0|$ the total number of vertices in $T_0$.
As $T_0$ is a $(C+1)$-ary tree, we have $(C+1) \cdot \ell \leq 1 + C \cdot |T_0|$.
Consider a leaf $v$ of $T_0$ which is not a leaf of $T$. 
It has $(C+1)$ children in $T$ and by definition, all are not in $U'_0$.
Thus, at most one set in $\mathcal{F}'_0$ can contain each of these children.
In particular, at most $(C+1)$ sets in $\mathcal{F}'_0$ contain $v$.
Consider a leaf $v$ of $T_0$ which is also a leaf of $T$.
There is only one root-to-leaf path containing $v$, and it appears in $\mathcal{F}_0$ with multiplicity $(C+1)$. Hence, there are at most $(C+1)$ sets in $\mathcal{F}_0'$ containing $v$.
From both cases we conclude that $|\mathcal{F}'_0| \leq (C+1)\cdot \ell \leq 1 + C \cdot |T_0|  \leq 1 + C \cdot |U'_0|$.
Thus, $|\mathcal{F}'_0| - C\cdot |U'_0| \leq 1$.
\end{proof}

We prove Theorem~\ref{removalupperbound} by taking $n$ disjoint copies of $\mathcal{F}_0$ over different base sets.

\subsection{Derandomizing Lemma~\ref{domset}}\label{derand}
The construction of a dominating set in Lemma~\ref{domset} can be done in a deterministic manner using the \emph{method of conditional expectations} (\cite{spencer1994ten} \cite{raghavan1988probabilistic}) as follows.
We first note that, in the notation of Lemma~\ref{domset}, and for every disjoint subsets $V_0,V_1\subseteq V(G)$, we can efficiently compute 
\[
E[|R|\;|\; \forall v\in V_1. v\in R_0\;\wedge\;\forall v\in V_0. v\notin R_0]
\]
by using the linearity of expectation and considering the following cases:
\begin{itemize}
    \item If $v\in V_1$ then $Pr(v\in R)=1$.
    \item If $v\in V_0$:
    \begin{itemize}
        \item If $N(v)\cap V_1 \neq \emptyset$ then $Pr(v\in R)=0$.
        \item Else, $Pr(v\in R)=(1-p)^{|N(v)\setminus V_0|}$.
    \end{itemize}
    \item If $v\notin V_0\cup V_1$:
    \begin{itemize}
        \item If $N(v)\cap V_1 \neq \emptyset$ then $Pr(v\in R)=p$.
        \item Else, $Pr(v\in R)=p + (1-p)^{1+|N(v)\setminus V_0|}$.
    \end{itemize}
\end{itemize}

We next notice that if $u\notin V_0\cup V_1$ then
\begin{align*}
E[|R|\;|\; \forall v\in V_1. v\in R_0\;\wedge\;\forall v\in V_0. v\notin R_0] =
p\cdot &E[|R|\;|\; \forall v\in V_1\cup\{u\}. v\in R_0\;\wedge\;\forall v\in V_0. v\notin R_0]
\\ +
(1-p)\cdot &E[|R|\;|\; \forall v\in V_1. v\in R_0\;\wedge\;\forall v\in V_0\cup\{u\}. v\notin R_0].
\end{align*}

Thus, we can add $u$ into either $V_0$ or $V_1$ without increasing the above expectation.
We therefore can iteratively add every vertex of $V(G)$ to either $V_0$ or $V_1$ without increasing the conditional expectation. 
We finish with a concrete choice of $R_0$ such that $|R|$ is bounded by the original $E[|R|]$.

\end{document}